\documentclass[journal, twocolumn]{IEEEtran}

\usepackage{comment,cite}
\usepackage{amsmath, amssymb, bm, bbm}
\usepackage{balance}
\usepackage{microtype}
\usepackage{graphicx}
\usepackage{subfigure}
\usepackage{booktabs} 
\usepackage{multirow}
\usepackage[colorlinks,pdfstartview=FitH,citecolor=blue,linkcolor=blue,urlcolor=blue]{hyperref}
\usepackage{makecell}
\usepackage{tabulary}

\newcommand\xqed[1]{%
  \leavevmode\unskip\penalty9999 \hbox{}\nobreak\hfill
  \quad\hbox{#1}}
\newcommand\closeremark{\xqed{$\triangle$}}

\usepackage{mathtools}

\DeclarePairedDelimiter\floor{\lfloor}{\rfloor}

\newcommand{\myapprox}{{\raise.17ex\hbox{$\scriptstyle\sim$}}}
\newcommand\set[1]{\left\{#1\right\}}

\newcommand\card[1]{\left\lvert#1\right\rvert}

\newcommand{\icol}[1]{
  \left[\begin{matrix}#1\end{matrix}\right]%
}

\newcommand\mydots{\hbox to 1em{.\hss.\hss.}}

\newcommand*\diff{\mathop{}\!\mathrm{d}}

\usepackage{xcolor}
\usepackage{braket, amsthm, amsfonts}
\usepackage[ruled,longend]{algorithm2e}
\usepackage{hyperref}

\usepackage{multirow}
\usepackage{enumitem} 

\newcommand{\E}{\mathbb{E}}

\usepackage[font=normalsize,skip=0pt]{caption}

\makeatletter
\g@addto@macro\normalsize{%
  \setlength\abovedisplayskip{5pt}
  \setlength\belowdisplayskip{5pt}
  \setlength\abovedisplayshortskip{5pt}
  \setlength\belowdisplayshortskip{5pt}
}
\makeatother

\makeatletter
\def\thm@space@setup{%
  \thm@preskip=0.15cm plus 0.05cm minus 0.05cm
  \thm@postskip=0.05cm plus 0.05cm minus 0.05cm
}
\makeatother

\newcommand{\Prob}[1]{\Pr\left\{#1\right\}}
\newcommand{\Exp}[1]{\mathbb{E}\left[#1\right]}
\newcommand{\I}[1]{\mathbbm{1}\left(#1\right)}

\theoremstyle{plain}
\newtheorem{theorem}{Theorem}
\newtheorem{lemma}{Lemma}

\newtheorem{corollary}{Corollary}

\theoremstyle{definition}
\newtheorem{definition}{Definition}

\theoremstyle{remark}
\newtheorem{remark}{Remark}

\allowdisplaybreaks
\makeatletter
\newcounter{longaligned}

\makeatother

\begin{document}
\title{Controlling Data Access Load in Distributed Systems}
\author{%
        Mehmet~F.~Akta\c{s},
        Emina~Soljanin,~\IEEEmembership{Fellow,~IEEE}
        
\thanks{M. F. Akta\c{s} and E. Soljanin are with the Department of Electrical and Computer Engineering, Rutgers University, New Brunswick,
NJ, 08901 USA. (email: \{mehmet.aktas, emina.soljanin\}@rutgers.edu)}
}

\maketitle

\begin{abstract}
  Distributed systems store data objects redundantly to balance the data access load over multiple nodes. Load balancing performance depends mainly on 1) the level of storage redundancy and 2) the assignment of data objects to storage nodes.
  We analyze the performance implications of these design choices by considering four practical storage schemes that we refer to as \emph{clustering}, \emph{cyclic}, \emph{block} and \emph{random} design. We formulate the problem of load balancing as maintaining the load on any node below a given threshold.
  Regarding the level of redundancy, we find that the desired load balance can be achieved in a system of $n$ nodes only if the replication factor $d = \Omega(\log(n)^{1/3})$, which is a necessary condition for any storage design. For clustering and cyclic designs, $d = \Omega(\log(n))$ is necessary and sufficient. For block and random designs, $d = \Omega(\log(n))$ is sufficient but unnecessary. 
  Whether $d = \Omega(\log(n)^{1/3})$ is sufficient remains open. 
  The assignment of objects to nodes essentially determines which objects share the access capacity on each node. We refer to the number of nodes jointly shared by a set of objects as the \emph{overlap} between those objects.
  We find that many consistently slight overlaps between the objects (block, random) are better than few but occasionally significant overlaps (clustering, cyclic). However, when the demand is ``skewed beyond a level'' the impact of overlaps becomes the opposite.
  We derive our main results by connecting the load-balancing problem to mathematical constructs that have been used in the literature to study other problems.
  For a class of storage designs containing the clustering and cyclic design, we express load balance in terms of the maximum of moving sums of i.i.d. random variables, which is also known as the \emph{scan statistic}. For random design, we express load balance using the \emph{occupancy metric for random allocation with complexes}.
    
\end{abstract}


\section{Introduction}
\label{sec:intro}
\noindent
\textbf{Motivation:}\space
Data access times are the main bottleneck to the performance of computing systems \cite{ChallengesInBuildingLargeScaleInformationRetrievalSystems:Dean09}.
In modern, large-scale cloud systems, data access times greatly suffer when storage nodes exhibit poor or variable performance \cite{Dremel:MelnikGL10}. Many factors cause poor performance, but primarily, it comes from resource sharing across multiple workloads. The resulting contention at the system resources creates overloaded storage nodes \cite{TailAtScale:DeanB13,StragglerRootCauseAnalysisInDatacenters:OuyangGY16}.
Therefore, distributed systems must be able to limit and control data access load at the storage nodes.

Offered load must be balanced across the storage nodes as evenly as possible.
Modern storage systems, such as HDFS \cite{HDFS:ShvachkoKR10}, Cassandra \cite{Cassandra:LakshmanM10}, and Redis \cite{Redis}, replicate data objects across multiple nodes to enable multiple \emph{service choices}. This storage redundancy allows splitting the offered load, which we refer to as the \emph{object demands}, across multiple nodes (service choices).
Storing each object at every node achieves the best support for load assignment, but it is costly and thus applied only when the system is sufficiently small.
Replicating objects with adequate redundancy requires knowledge of fixed object demands, but in practice, object demands are unknown and fluctuate over time. The design of storage redundancy is the first and arguably the most critical step in achieving robust load balancing in the presence of skews and changes in object popularities \cite{ChallengesInBuildingLargeScaleInformationRetrievalSystems:Dean09, Scarlett:AnanthanarayananAK11}.

Large-scale systems strive to balance the offered load over the nodes by using as minimal storage redundancy as possible.
Load balancing performance is determined by (1) the level of redundancy for each data object and (2) the assignment of objects to nodes. A higher level of storage redundancy implies better load balancing, but it also incurs higher costs. The impact of object-to-node assignment is more subtle. Depending on the object demands, the impact of object-to-node assignment on load balancing can be significant or negligible.
In this paper, we characterize both factors' impact on load-balancing performance.

\vspace{1ex}
\noindent
\textbf{Prior and Related work}:\space
Balancing the offered load has been studied in two critical settings. 
In the first, which we call the \emph{dynamic setting}, load balancing has been studied in scheduling tasks to compute nodes. Each node is assumed to serve tasks through a first-in-first-out queue. Tasks are balanced by querying the queue lengths at a subset of the nodes and assigning the task to the least loaded one.
Ideal load balance is achieved by querying all the nodes for each task arrival, which is impractical in large-scale systems.
For this reason, much research has gone into developing techniques that query a limited number of nodes for task assignment. These techniques are based on the well-known \emph{power of $d$ choices} paradigm.
The research work following this direction has produced a plethora of asymptotic results on the system performance, often via analysis using the balls-into-bins models \cite{BallsIntoBins:RaabS98, BalancedAllocations:AzarBK99, BalancedAllocations_HeavilyLoadedCase:Berenbrink20}.

The above literature only applies to systems that can direct any demand to any node, e.g., scheduling tasks to compute nodes. 
However, in storage systems, we do not have the flexibility of querying any subset of nodes for scheduling because an arriving data access request can only be served at one of the nodes storing the requested object.
A more appropriate model for storage systems would assume that requests can be offered to a limited collection of subsets of nodes.
A model along these lines has been considered in \cite{BalancedAllocationsOnGraphs:Kenthapadi06} where subsets of nodes are represented as edges in a graph. The paper studied the 
power of two choices paradigm on this restricted model for assigning $n$ balls sequentially to $n$ bins.
This model with a graph was extended to one with a hypergraph in \cite{BalancedAllocationsOnHypergraphs:Godfrey08}, which allows for studying the general power of the $d$ choices paradigm.
Storage allocations in this paper are special cases of the balanced allocations on hypergraphs considered in \cite{BalancedAllocationsOnHypergraphs:Godfrey08}. The results presented in \cite{BalancedAllocationsOnHypergraphs:Godfrey08} provide limited insight into practical storage schemes, for example, for the impact of storage allocations on the load balance. These results are also shown only in the lightly loaded case when the cumulative load offered on the system scales as the order of the number of nodes.
This paper examines essential storage properties not addressed in \cite{BalancedAllocationsOnHypergraphs:Godfrey08} without restrictions to the lightly loaded case. These properties include the number of different objects stored per node and object overlaps between the storage nodes.

In the dynamic setting, data access requests arrive sequentially and get assigned to nodes upon arrival. Hence, the load offered for the objects is not known a priori.
In the second setting, which we name the \emph{static setting}, the goal is to answer a different question: can the system serve the offered load if the object demands are known from the start? The load balance achieved in the dynamic setting is always achievable in the static setting thanks to knowing object demands a priori.
Data access performance in the static setting, therefore, represents the best-case performance of the system.

There are two distinct approaches in the static setting.
In the first approach, the goal is to serve the load as long as any $m$ objects chosen with replacement out of all stored objects are requested simultaneously, which leads to the design of \emph{batch codes} \cite{BatchCodesAndTheirApps:IshaiKO04}.
The storage schemes we consider here fall into the class of combinatorial batch codes \cite{CombinatorialBatchCodes:StinsonWP09}.
In the second approach, the offered load model is extended. The goal is to find the system's \emph{service capacity region}, which is defined as a set of all demand vectors that a system can serve with a given storage scheme \cite{AllertonServiceCapacity:AktasJS17, ServiceRateRegion:AktasJK21}.
Our treatment of load balancing falls into this second approach.

Capacity region for various storage schemes has been derived in \cite{AllertonServiceCapacity:AktasJS17} and \cite{ITWServiceCapacity:AndersonJJ18}. This line of work, however, considers only the case where each node stores a single object.
Most importantly, even though capacity region gives a sense of the system's overall capacity to deal with changes in the offered load, it does not capture the probability of stable service when the expected load is random.
This paper addresses this gap by analyzing the system's load balancing performance with a stochastic offered load model and a probabilistic performance metric.
We focus our analysis on four storage designs with object replication commonly considered or deployed in practical systems to achieve high data availability.

Similar to this work, a stochastic formulation has been proposed recently in \cite{LoadBalancing:AktasFS21} to analyze load-balancing performance in systems with multiple objects per node. Using this formulation, the authors have drawn various conclusions on the design of storage schemes. However, the results presented in \cite{LoadBalancing:AktasFS21} are all asymptotic as the scale of the system goes to infinity. These results provide insight into the performance improvement achieved by increasing the replication factor; they do not provide insight into the performance impact of different storage allocations. Also, the offered load model assumed in \cite{LoadBalancing:AktasFS21} captures only a particular set of load characteristics and is not extensible. 
We propose an offered load model that permits modifications for capturing different load characteristics.
Furthermore, our performance analysis is on the finite case and sheds light on the replication factor and the performance impact of different storage designs.

\vspace{1ex}
\noindent
\textbf{Contributions and Organization}:\space
We propose a stochastic offered load model that captures the fluctuations and skews in object demands.
Our load model can be modified to capture different demand characteristics by changing the \emph{object demand distribution}.
We use a probabilistic metric to measure the system's ability to limit the load on the maximally loaded storage node below a threshold, which we refer to as the \emph{maximal access load} requirement. Specifically, we define the performance metric as the probability $\mathcal{P}$ of meeting the maximal access load requirement for object demands randomly sampled from the offered load model.
We then present a mathematical analysis of $\mathcal{P}$ for systems that store each object with the same replication factor, hence the name \emph{regular} redundancy. There are many ways to allocate object replicas to storage nodes. We consider four different storage allocation strategies that are used in practical systems. We refer to them as \emph{clustering}, \emph{cyclic}, \emph{block}, and \emph{random} design.

We refer to the nodes storing an object as its \emph{service choices}. We refer to the intersection of service choices for a set of objects as their \emph{service choice overlaps}.
We show that the cumulative service choice overlaps remain fixed within the class of storage allocations that assign the same number of object copies per node (\emph{balanced} redundancy). Thus, within the class of regular and balanced storage allocations, the specific allocation design determines the distribution of service choice overlaps across the subsets of objects.
The service choice overlaps are highly skewed by clustering design and evenly distributed by block design. Cyclic and random design lie between the two.

We show that analyzing $\mathcal{P}$ exactly requires taking the demand assignment for all objects into account at once, and explain why this is intractable for an arbitrary storage allocation. We show that it is still possible to derive an upper bound on $\mathcal{P}$ for any given storage allocation, which is tight for the allocation with the best performance.
This, however, is not accurate enough to yield insight on the performance impact of different allocation designs.
We analyze $\mathcal{P}$ for different designs using ideas and results developed for different problems in applied probability.
The structure of clustering design allows deriving exact expressions for $\mathcal{P}$. The cyclic or random design does not permit this. 
We show that $\mathcal{P}$ for cyclic design is connected to the distribution of \emph{scan statistic} \cite{ScanStats:GlazNW01}. We also show that $\mathcal{P}$ for random design is connected to the occupancy distribution for \emph{random allocation with simplexes} \cite{RandomAllocations:KolchinSC78}. We use these connections to find upper and lower bounds on $\mathcal{P}$ for cyclic, block and random design.

The bounds we find on $\mathcal{P}$ show how the system performance depends on the important system parameters such as the number of objects $k$ stored in the system, the number of storage nodes $n$, and the replication factor $d$.
These bounds also allow us to derive the behavior of $\mathcal{P}$ as the system scale $n$ gets larger.
First, we show that $\mathcal{P} \to 0$ as $n \to \infty$ unless replication factor $d = \Omega(\log(n)^{1/3})$, regardless of which storage design is used.
We then show conditions for $\mathcal{P} \to 1$ as $n \to \infty$ as follows.
For clustering and cyclic design, $d = \Omega(\log(n))$ is necessary and sufficient.
For block and random design, $d = \Omega(\log(n))$ is sufficient but not necessary. That is, $\mathcal{P}$ can possibly converge to $1$ even when $d = \Omega(\log(n)^{1/3})$. However, we do not have a proof for this and leave it as a future work.
Note that, Table~\ref{table:summary} below summarizes the main findings presented in this paper.

This paper is organized as follows:
Sec.~\ref{sec:intro} gives an overview of the literature on analyzing data access performance for storage systems. We also discuss the connections between our approach and the prior work.
Sec.~\ref{sec:preliminaries} presents our storage and offered load model. 
In Sec.~\ref{sec:perf_metric}, we define the metrics we use to evaluate system performance.
In Sec.~\ref{sec:1_choice}, we consider storage allocations with no redundancy and evaluate the performance impact of the number of objects stored per node.
In Sec.~\ref{sec:d_choice}, we consider four different practical storage allocation strategies with object replication and analyze their performance.

\begin{table*}
\caption{Summary of our results on the load balancing performance for four different object-to-node assignment strategies.}
\vspace{1ex}
\centering
\begin{tabulary}{1.0\textwidth}{LL|CCC|p{.1\textwidth}}
\hline
\multicolumn{2}{|l|}{Object-to-node assignment strategy
}                                                                                                            & \multicolumn{1}{c|}{$\qquad\;\;$Clustering$\qquad\;\;$}                                                                & \multicolumn{1}{c|}{$\qquad\;\;$Cyclic$\qquad\;\;$} & \multicolumn{1}{c|}{$\qquad\;\;$Random$\qquad\;\;$} & \multicolumn{1}{c|}{Block}                                                                                  \\ \hline
\multicolumn{2}{|l|}{Distribution of service choice overlaps}                                                                                           & \multicolumn{4}{l|}{$\quad$ Heavily skewed $\quad\qquad\qquad\qquad\qquad\qquad \longrightarrow \quad\qquad\qquad\qquad\qquad\qquad$ Evenly distributed} \\ \hline
\multicolumn{2}{|l|}{\makecell[l]{Bounds on replication factor $d$ \\ to meet the maximal load requirement} }                            & \multicolumn{2}{c|}{\begin{tabular}[c]{@{}c@{}}Necessary and sufficient:\\ $d = \Omega(\log(n))$\end{tabular}}               & \multicolumn{2}{c|}{\begin{tabular}[c]{@{}c@{}}Necessary: $d = \Omega(\log(n)^{1/3})$\\ Sufficient: $d = \Omega(\log(n))$\end{tabular}}   \\ \hline
\multicolumn{1}{|c|}{\multirow{2}{*}{\makecell[l]{Load balancing \\ performance comparison}}} & \makecell[l]{when object demands \\ are heavily skewed} & \multicolumn{4}{c|}{\begin{tabular}[c]{@{}c@{}}$\mathcal{P}_{\mathrm{clustering}} > \mathcal{P}_{\mathrm{cyclic}} > \mathcal{P}_{\mathrm{random}} > \mathcal{P}_{\mathrm{block}}$\\ (skewed overlaps is better)\end{tabular}}                                            \\ \cline{2-6} 
\multicolumn{1}{|c|}{}                             & \multicolumn{1}{c|}{otherwise}                                                     & \multicolumn{4}{c|}{\begin{tabular}[c]{@{}c@{}}$\mathcal{P}_{\mathrm{clustering}} < \mathcal{P}_{\mathrm{cyclic}} < \mathcal{P}_{\mathrm{random}} < \mathcal{P}_{\mathrm{block}}$\\ (evenly distributed overlaps is better)\end{tabular}}                                              \\ \hline
\multicolumn{6}{c}{
\makecell[l]{
The system consists of $n$ nodes and stores $d$ copies for each data object -- see Sec.~\ref{subsec:sys_model} for the details on system model. Load balancing performance is \\ 
measured with metric $\mathcal{P}$ with object-to-node assignment indicated with subscript -- see Sec.~\ref{subsec:data_access_performance} for the details on performance metric.}}
\end{tabulary}
\label{table:summary}
\end{table*}

\section{System Model}
\label{sec:preliminaries}
We study data access in the static setting with continuous service and offered load model.
Our model reveals a connection between the data access problem, convex polytopes, \emph{scan statistics} \cite{ScanStats:GlazNW01} and \emph{random allocations with complexes} \cite{RandomAllocations:KolchinSC78}. 

\subsection{Storage and Access Model}
\label{subsec:sys_model}
We consider a system of $n$ storage nodes hosting $k$ data objects $o_1, \dots, o_k$, and their replicas.
Nodes have identical access \emph{capacity}, defined as the maximum number of bytes a node can stream per second.
An \emph{object} denotes the smallest data unit as a fixed-length string of bits.

We refer to the \emph{offered load} for object $o_i$ as its \emph{demand} $\rho_i$.
Demand for an object represents the average number of bytes streamed from the system per second to access the object, divided by a single node's access capacity.
We refer to a node that hosts an object as a \emph{service choice} for the object. Replicating an object $o_i$ over multiple nodes creates a set of multiple service choices $C_i$.


Demand for an object can be split arbitrarily across its service choices. 
The \emph{load assigned on a node} equals the sum of the offered load portions exerted on it by the objects stored on it.
A node is said to be \emph{stable} if its assigned load is less than $1$. A system is said to be stable if every node is stable.
We define the \emph{maximal load} as the maximum load across all the nodes.
We assume that, if feasible, the object demands $\rho_i$ are split across their service choices so that the maximal load is below a given threshold $m \in (0, 1]$. It may be unfeasible to keep the maximal load below a given $m$ depending on the object demands and the storage allocation.


As we discuss in more detail in Sec.~\ref{subsec:cap_region}, whether the system can achieve the desired maximal load or not can be determined by (1) solving a linear program or (2) checking a set of conditions on the union of objects' service choices. We refer to the union of service choices for a set of objects as the \emph{span} of these objects. Notice that the span of a set of objects equals the total capacity available to serve those objects jointly.
\begin{definition}
  The \underline{service choice span} for the set of objects $O = \{o_i ~\mid~ i \in I \subset \{1, \dots, k\}\}$ is given as
  \[ 
      \mathrm{span}(O) = \card{\bigcup_{i \in I} C_i}.
  \]
 Note that the span of a single object $o_i$ is given by $C_i$.
\label{def:service_choice_span}
\end{definition}

A \emph{storage allocation} determines how objects are assigned to the storage nodes.
This paper focuses on \emph{regular balanced $d$-choice} storage allocations.
\begin{definition}
  A \underline{regular balanced $d$-choice allocation} stores each object with $d$ service choices and distributes object copies across the nodes so that each node stores the same number of different objects.
\label{def:reg_balanced_dchoice_alloc}
\end{definition}

In the rest of the paper, unless otherwise noted, the allocation itself will refer to a regular balanced allocation.
There are many ways to design a $d$-choice allocation. We detail the allocations we consider in Sec.~\ref{sec:d_choice}.

\subsection{Offered Load Model}
\label{subsec:offered_load}

In practical systems, object demands change depending on many factors, such as time of the day, cumulative load offered on the system, and object popularity.
We model the object demands $\rho_i$ as i.i.d. non-negative random variables. We capture different offered load characteristics using different demand distributions.
Some of the demand distributions we consider are 
(1) \emph{Exponential} with rate $\mu$, i.e., $\rho_i \sim \mathrm{Exp}(\mu)$.
(2) \emph{Pareto} with minimum value $\lambda$ and tail index $\alpha$, i.e., $\rho_i \sim \mathrm{Pareto}(\lambda, \alpha)$.
(3) \emph{Scaled Bernoulli} with constant scale $\lambda$ and probability of success $p$, i.e., $\rho_i \sim \lambda \times \mathrm{Bernoulli}(p)$.

\subsection{Note on the Proofs and the Notation}
We denote the distribution of the sum of $u$ (i.i.d.) object demands $\Prob{\rho_1 + \dots + \rho_u \leq x}$ with $F_u(x)$.
$F(x)$ denotes the distribution of a single object demand.
We use $\phi_X(t)$ to denote the moment generating function $\Exp{\exp(tX)}$ for a random variable $X$.

We place the proofs in the Appendix.
Throughout the paper, $\log$ refers to the natural logarithm.
We use the ``$\to$'' notation to denote the convergence of a sequence.
Let $\{f_n(x); \;n \geq 1\}$ be a sequence of functions $f_n : D \to \mathbb{R}$. Let $f(x)$ be another function $f : D \to \mathbb{R}$.
If $\lim_{n \to \infty} f_n(x) = f(x)$ at every $x \in D$, we will denote this as $f_n(x) \to f(x)$.

\section{Storage Service Performance}
\label{sec:perf_metric}

Our performance metric is the probability of the system being able to serve the offered data access load while meeting the maximum load requirement. We first introduce the ``service capacity region'' for a given storage system and then use it to define our performance metric.

\subsection{Storage Service Capacity}
\label{subsec:cap_region}
\emph{Capacity region} for a given storage system consists of all the object demand vectors the system can serve. The systems of interest, in this case, store data objects with redundancy.
The concept of service capacity region for storage systems has been introduced in \cite{AllertonServiceCapacity:AktasJS17}, and further studied in \cite{ITWServiceCapacity:AndersonJJ18}.
\begin{definition}
  \underline{Capacity region} for a system with a given storage allocation is the set of all object demand vectors $(\rho_1, \dots, \rho_k)$ that the system can serve while operating under stability.
 \label{def:cap_region}
\end{definition}

We limit the capacity region's presentation to the minimum needed to understand the definitions given in this paper.
We refer the reader to \cite[Sec.~2-C]{LoadBalancing:AktasFS21} for a concise presentation of the formulation and to \cite{ServiceRateRegion:AktasJK21} for full exposure to the subject. The capacity region of a system is the following \emph{convex polytope}:
\begin{equation}
  \mathcal{C} = \left\{\bm{\rho} ~\mid~ \exists \bm{x}; ~\bm{M} \cdot \bm{x} \prec \bm{1}, ~\bm{T} \cdot \bm{x} = \bm{\rho}, ~\bm{x} \succeq \bm{0} \right\}.
\label{eq:cap_region_w_matrix_inequalities}
\end{equation}
where $\bm{M}$ is a binary matrix expressing the storage allocation, $\bm{1}$ is the all ones vector, $\bm{x}$ is a vector of real numbers and $\bm{T}$ is a binary matrix that transforms a given $\bm{x}$ to the corresponding demand vector $\bm{\rho}$.

We now introduce the \emph{modified capacity region} $\mathcal{C}_m$ for a given maximal load $m$. This notion will help us define the metric we use to measure data access performance.
We define $\mathcal{C}_m$ as the set of demand vectors the system can serve while keeping the maximal load below $m$. It is expressed as
\begin{equation}
  \mathcal{C}_m = \left\{\bm{\rho} ~\mid~ \exists \bm{x}; ~\bm{M} \cdot \bm{x} \prec m\bm{1}, ~\bm{T} \cdot \bm{x} = \bm{\rho}, ~\bm{x} \succeq \bm{0} \right\}.
\label{eq:modified_cap_region_w_matrix_inequalities}
\end{equation}

The following subsection presents a more insightful expression of the capacity region. We will extend this expression for the modified capacity region $\mathcal{C}_m$, giving us a more analytically helpful way to calculate our performance metric.

\subsection{Capacity region in terms of service choices}
\label{subsec:cap_region_via_service_choice_spans}

Suppose that for each object $o_i$, the system stores $d_i$ copies distributed across different nodes. In other words, object $i$ has $d_i$ service choices, which we denote with set $C_i$.
Notice that this storage allocation is more general than the regular balanced allocation described in Sec.~\ref{subsec:sys_model}.
For such systems with replicated storage redundancy, we can express the capacity region by designating a lower bound for the service choice spans (see Def.~\ref{def:service_choice_span}) for every subset of objects.
We present the expression in the following Lemma.
\begin{lemma}
  The capacity region of a system that stores $d_i$ copies for each object $o_i$ is given by
  \begin{equation}
    \mathcal{C} = \Bigl\{\bm{\rho} \;\Bigl|\; \sum_{i \in I}{\rho_i} \leq \mathrm{span}(o_i; \;i \in I), ~\forall I \subset \{1, \ldots, k\}\Bigr\}.
  \label{eq:cap_region_w_service_choice_spans}
  \end{equation}
\label{lm:cap_region_for_system_w_replicated_storage}
\end{lemma}
\begin{proof}
  See Appendix~\ref{subsec:proof_lm_cap_region_for_replicated_storage}.
\end{proof}

We can extend \eqref{eq:cap_region_w_service_choice_spans} to express the modified capacity region $\mathcal{C}_m$ (defined in \eqref{eq:modified_cap_region_w_matrix_inequalities} as) that contains the set of demand vectors under which the system can meet the maximal load requirement
\begin{equation}
  \mathcal{C}_m = \Bigl\{\bm{\rho} \;\Bigl|\; \sum_{i \in I}{\rho_i} \leq m \cdot \mathrm{span}(o_i; \;i \in I), ~\forall I \subset \{1, ., k\}\Bigr\}.
\label{eq:modified_cap_region_w_service_choice_spans}
\end{equation}

We rely on \eqref{eq:modified_cap_region_w_service_choice_spans} Sec.~\ref{sec:d_choice} where we analyze $\mathcal{P}$ for systems with different storage allocations.

\subsection{Data Access Performance}
\label{subsec:data_access_performance}
We measure data access performance with the system's \emph{robustness} against the changes in object demands. A robust storage system should be able to maintain the maximal access load below a desired level $m$ when the object demands change.
We quantify system robustness as the probability that the system can serve the offered demand vector while keeping the maximal load below $m$.
In our offered load model (discussed in Sec.~\ref{subsec:offered_load}), the set of demand vectors that can be offered on the system and their likelihood is determined by the demand distribution $F(x)$.

\begin{definition}
  For a system with a given storage allocation, let $\mathcal{C}_m$ denote the modified capacity region with maximal load $m$ as defined in \eqref{eq:modified_cap_region_w_matrix_inequalities}.
  \underline{Robustness $\mathcal{P}$} for the system denotes the probability that the system can serve the offered demand vector while keeping the maximum load below $m$
  \begin{equation}
    \mathcal{P} = \Prob{\bm{\rho} \in \mathcal{C}_m}.
  \label{eq:P}
  \end{equation}
\label{def:P}
\end{definition}

The expression in \eqref{eq:P_integration} is a useful geometric interpretation for $\mathcal{P}$.
It implies that once the modified capacity region $\mathcal{C}_m$ of a system is determined, $\mathcal{P}$ can be evaluated by integrating the  joint density of the object demands over $\mathcal{C}_m$ as
\begin{equation}
  \mathcal{P} = \int_{\bm{c} \in \mathcal{C}_m} \Prob{\bm{\rho} = \bm{c}} \diff \bm{c}.
\label{eq:P_integration}
\end{equation}
Recall from Sec.~\ref{subsec:offered_load} that the object demands $\rho_i$ are sampled independently from the same distribution with 
density function $f(x)$. Thus we can write the integrand as
\begin{equation}
  \Prob{\bm{\rho} = \bm{c}} = \prod_{i = 1}^{k} f(c_i).
\label{eq:P_integrand}
\end{equation}

As the service capacity region is a convex polytope, \eqref{eq:P_integration} implies that $\mathcal{P}$ is non-increasing in demand distribution $F(x)$.
\begin{corollary}
  For a system with any storage allocation, let the system robustness be $\mathcal{P}$ and $\mathcal{P}^\prime$ for demand distributions $F(x)$ and $F^\prime(x)$ respectively. If $F^\prime(x) < F(x)$, then $\mathcal{P}^\prime < \mathcal{P}$.
\label{cor:P_larger_demand_leq_P}
\end{corollary}
\begin{proof}
  See Appendix~\ref{subsec:proof_cor_P_larger_demand_leq_P}.
\end{proof}

Given that $\mathcal{C}_m$ is a convex polytope, we can check whether a demand vector is in $\mathcal{C}_m$ by solving a linear feasibility problem. This solution, together with \eqref{eq:P_integration} gives a recipe to exactly compute $\mathcal{P}$ for systems with any storage allocation. However, solving a linear program may not help understanding what differentiates a storage allocation with good performance from those with low performance.

\section{Storage with No Redundancy}
\label{sec:1_choice}

We consider \emph{single-choice} allocations in which each of the $k$ objects is stored on only a single node, and each of the $n$ nodes stores $b=k/n$ different objects, where we assume $n|k$.
Demand for each object, in this case, must be served entirely by the only node hosting the object, and each node has to serve the total demand for all objects stored on it.
This straightforward assignment of object demands to nodes makes it possible to derive an exact expression for $\mathcal{P}$.
\begin{lemma}
  In a single-choice storage allocation
  \begin{equation}
    \mathcal{P} = F_b(m)^{n}.
  \label{eq:P_single_choice}
  \end{equation}
\label{lm:P_single_choice}
\end{lemma}

When distribution $F_b(x)$ of the sum $\rho_1 + \ldots + \rho_b$ has a closed form expression, \eqref{eq:P_single_choice} would give a closed form expression for $\mathcal{P}$. For instance, if $\rho_i \sim \mathrm{Exp}(\mu)$, then \eqref{eq:P_single_choice} would be given as
\[
    \mathcal{P} = \left( \frac{\gamma(b, m \mu)}{(b - 1)!} \right)^{n},
\]
where $\gamma$ is the lower incomplete gamma function.

The expression \eqref{eq:P_single_choice} enables us to evaluate the impact of $b$ on $\mathcal{P}$. Let us now assume that the object demands are distributed as $\rho_i / b$. This ensures a fair comparison between different values of $b$ by keeping the average cumulative demand on the system fixed.
With this assumption, \eqref{eq:P_single_choice} becomes
\begin{equation*}
    \mathcal{P} = \Prob{(\rho_1 + \ldots + \rho_b) / b \leq m}^{n}.
\end{equation*}
When $\rho_i$ have a finite expected value $\mu$, the mean $(\rho_1 + \ldots + \rho_b) / b$ becomes more concentrated around $\mu$ with increasing $b$. For most distributions of $\rho_i$, the concentration of mean would imply that $\mathcal{P}$ increases with $b$.
This is shown in Fig.~\ref{fig:P_vs_b} for $\rho \sim \mathrm{Exp}$.
We cannot, however, generalize this for all demand distributions. For instance, if $\rho_i$ are distributed as Pareto with an infinite expected value, then the mean $(\rho_1 + \ldots + \rho_b) / b$ would be larger than $\rho_i$ in the sense of first-order stochastic dominance \cite{ParetoStochasticDominance:ChenEW22}.

As $\mathcal{P}$ increases with $b$ for most demand distributions of interest (while keeping the average cumulative demand fixed), it is safe to conclude that storing one object per node is the worst-case in terms of $\mathcal{P}$. From now on, we consider the worst-case $b = 1$, which implies $k = n$. Under this assumption, the system is configured with only two parameters: the replication factor $d$ and the \emph{system scale} $n$.
This assumption makes it more tractable to formulate and study the data access problem and easier to explain and interpret the presented results.
The results that we present can be extended for the general case with a fixed value of $k/n > 1$ by using arguments that are very similar to those we discuss.

\begin{figure}[t]
  \centering
    \includegraphics[width=.4\textwidth]{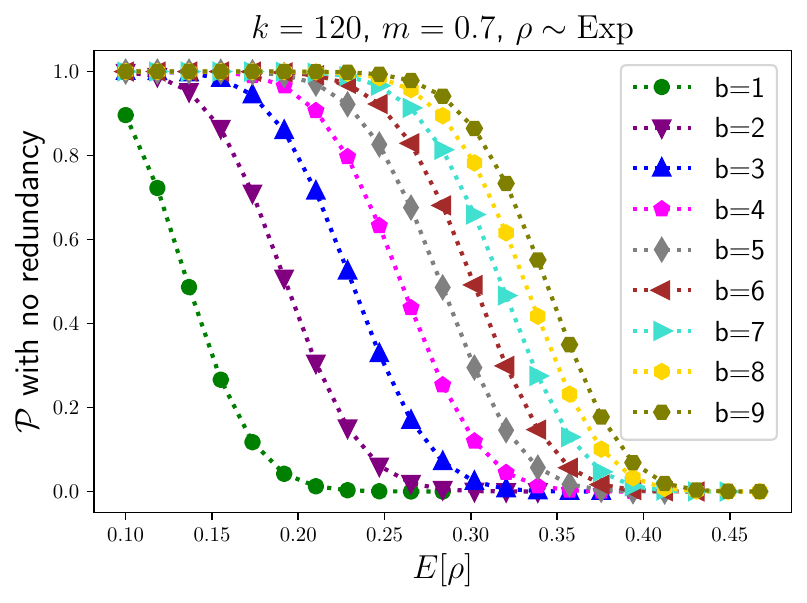}
  \caption{$\mathcal{P}$ vs average object demand for no-redundancy system with varying number $b$ of objects stored per node.}
\label{fig:P_vs_b}
\end{figure}

\section{Storage with Replication}
\label{sec:d_choice}
In this section, we consider \emph{$d$-choice} allocations in which each of the $n$ objects (recall our assumption $k = n$ from Sec.~\ref{sec:1_choice}) is stored on $d$ different nodes. In other words, we consider storage systems where each object has $d$ service choices. Our goal in this section is to understand the impact of the replication factor $d$ and the object-to-node assignment on data access performance in terms of $\mathcal{P}$.

A $d$-choice storage allocation defines a $d$-regular bipartite mapping from the set of objects to the set of nodes, which we refer to as the \emph{allocation graph}.
A $d$-choice allocation is constructed as follows: for each object (1) select a set of $d$ nodes out of the total $n$ nodes according to some object-to-node assignment strategy, (2) store the object on the selected nodes. We denote the set of nodes that host object $o_i$ with $C_i$. In other words, $C_i$ consists of the service choices for $o_i$.
The object-to-node assignment strategy determines the set of service choices $C_i$ for the objects.

\subsection{Service Choice Spans and Overlaps}
\label{subsec:service_choice_spans}

In a $d$-choice allocation, the total capacity to serve each object individually is $d$. The capacity to jointly serve multiple objects is given by their span (see Def.~\ref{def:service_choice_span}). Recall from Lemma~\ref{lm:cap_region_for_system_w_replicated_storage} that a demand vector lies within the system's capacity region if the service choice spans are large enough to meet the cumulative demand for all subsets of objects.
It is then desirable to maximize the service choice spans to expand the capacity region and cover more demand vectors. For instance, Fig.~\ref{fig:cap_region_for_two_objs} shows how the capacity region for two objects shrinks as their span is reduced. Span (union) is inversely proportional to the size of the overlap (intersection) between the service choices. Therefore, one should reduce the overlaps to expand the capacity region. However, simultaneously reducing the service choice overlaps for all subsets of objects is impossible. For instance, suppose we move a copy of object-$i$ from node-$u$ to node-$v$ to reduce the overlap between object-$i$ and the other objects stored on node-$u$. This move will increase the overlap between object-$i$ and the objects previously stored on node-$v$.

\begin{figure*}[t]
  \centering
    \includegraphics[width=.9\textwidth]{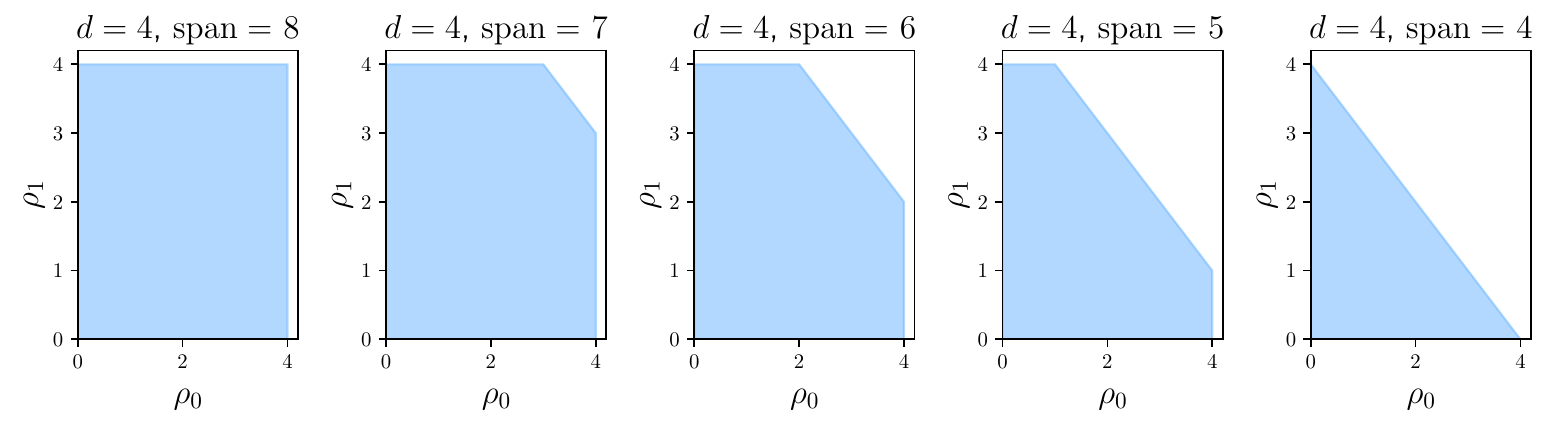}
  \caption{Capacity region for objects $o_0$ and $o_1$ as their span goes from maximum to minimum.}
\label{fig:cap_region_for_two_objs}
\end{figure*}

Although it is impossible to simultaneously reduce the overlaps between all subsets of objects, as we show below, we can reduce the cumulative overlap between all $t$-subsets of objects where $t > 1$. We next define the cumulative overlap.
\begin{definition}
    For a storage allocation with $n$ objects, the cumulative overlap between the service choices $C_i$ of $t$-subsets of objects is defined as
    \begin{equation}
    \mathrm{CumOverlap}_t = \sum_{\{i_1, \ldots, i_t\} \subset \{1, \dots, n\}} \left| C_{i_1} \cap \ldots \cap C_{i_t} \right|.
  \label{eq:cum_overlap_t_subsets}
  \end{equation}
\label{def:cum_overlap}
\end{definition}

Note that cumulative overlaps directly give us the cumulative service choice spans. For instance, the cumulative span for object pairs is given by
\begin{equation}
    \sum_{\{i_1, i_2\} \subset \{1, \dots, n\}} \left| C_{i_1} \cup C_{i_2} \right|.
\label{eq:cum_span}
\end{equation}
We can express the service choice span for two objects in terms of the cardinality of their overlap as
\[
    \left| C_{i_1} \cup C_{i_2} \right| = \left| C_{i_1} \right| + \left| C_{i_2} \right| - \left| C_{i_1} \cap C_{i_2} \right|,
\]
where the span of a single object is $d$, i.e., $|C_{i_1}| = |C_{i_2}| = d$.
Substituting \eqref{eq:cum_overlap_t_subsets} into the above expression, we express cumulative span in terms of cumulative overlap as
\[
    {n \choose 2}2d - \mathrm{CumOverlap}_2.
\]
Similar to above, cumulative span for $t > 2$ can be written in terms of the cumulative overlaps using the following equality
\begin{equation}
\begin{split}
    \Bigl| C_{i_1} &\cup \ldots \cup C_{i_t} \Bigr| = \\
    &\sum_{u=1}^t (-1)^{u+1} \sum_{\{i_1, \ldots, i_u\} \subset \{1, \dots, n\}} \Bigl| C_{i_1} \cap \ldots \cap C_{i_u} \Bigr|.
\end{split}
\label{eq:cum_span}
\end{equation}

We next show that taking one step towards balancing the objects over the nodes reduces cumulative service choice overlaps.
The following Lemma shows that moving an object from a node with $u$ objects to another with fewer objects reduces the cumulative overlaps $\mathrm{CumOverlap}_t$ for $t < u$.
\begin{lemma}
    In a given storage allocation, moving an object from a node with $u$ objects to another node with $v$ objects where $u > v + 1$ reduces $\mathrm{CumOverlap}_t$ by ${u - 1 \choose t - 1} - {v \choose t - 1}$.
\label{lm:balancing_allocation_reduces_cum_overlap}
\end{lemma}
\begin{proof}
  See Appendix~\ref{subsec:proof_lm_balancing_allocation_reduces_cum_overlap}.
\end{proof}

Lemma~\ref{lm:balancing_allocation_reduces_cum_overlap} shows that balancing the allocation reduces the cumulative overlaps. This, however, does not always lead to a reduction in $\mathcal{P}$.
To see that, consider the following example.
Suppose an object is moved from node-$i$ to node-$j$, hosting fewer objects. Suppose the objects already stored on node $j$ overlap at all their service choices. Then, moving another object to node-$j$ will increase the competition for the capacity available on node-$j$. This will reduce $\mathcal{P}$, or at best keep it the same. The latter would happen if the objects stored on node-$i$ used to overlap at all their service choices, and the situation comprehensively improves with one of the objects moving to node-$j$.

Examples such as the one in the previous paragraph do not constitute most cases. Balancing the allocation is more likely to increase $\mathcal{P}$ than to reduce it. Moreover, reducing cumulative overlaps is still a reasonable goal. As discussed above, reducing the overlaps for a set of objects leads to a larger capacity for serving those objects jointly, which is more likely to increase $\mathcal{P}$ than to reduce it.
In the following sub-section, we discuss \emph{balanced} allocations and show that they minimize the cumulative service choice overlaps.

\subsection{Balanced Allocations}
\label{subsec:balanced_d_choice}

Recall from Def.~\ref{def:reg_balanced_dchoice_alloc} that we designate a storage allocation as \emph{balanced} when each node stores the same number of object copies.
Construction of a balanced $d$-choice allocation can be described as follows:
i) Map primary copies for all objects to nodes with a bijection $f_0$,
ii) For $i$ going from $1$ to $d$, map the $i$th redundant object copies to nodes with a bijection $f_i$ such that $f_i(o) \neq f_j(o)$ for every $j < i$ and $o$.
Thus, every node stores a single primary and $d-1$ redundant object copies, and each copy stored on the same node is for a different object.

The following Lemma shows that balancing $d$-choice allocations minimizes the cumulative service choice overlaps. This result, together with \eqref{eq:cum_span}, implies that the \emph{cumulative object span is maximized} in balanced allocations.
\begin{lemma}
  In a balanced $d$-choice allocation for storing $n$ objects, the cumulative overlap between $t$-subsets of objects (see Def.\ref{def:cum_overlap}) is given as
  \begin{equation}
    \mathrm{CumOverlap}_t = n {d \choose t},
  \label{eq:cum_overlap_in_balanced_d_choice}
  \end{equation}
  which is the minimum value across all $d$-choice allocations. 
\label{lm:balanced_d_choice_allocation_cum_overlap}
\end{lemma}
\begin{proof}
  See Appendix~\ref{subsec:proof_lm_balanced_d_choice_allocation_cum_overlap}.
\end{proof}

It is worth recalling that balanced $d$-choice allocations implement batch codes \cite{LoadBalancing:AktasFS21}. Batch codes constructed with replication are known as \emph{combinatorial} batch codes \cite{CombinatorialDesigns:Stinson07, CombinatorialBatchCodes:StinsonWP09}.
As shown in Lemma~\ref{lm:balanced_d_choice_allocation_cum_overlap}, cumulative overlaps are fixed in balanced $d$-choice allocations. However, multiple ways exist to distribute the overlaps across different subsets of objects via different allocation designs. Different designs favor different service choice overlap characteristics while yielding the same cumulative overlap. We will next discuss four designs for constructing $d$-choice storage allocations.

\vspace{1ex}
\noindent
\textbf{Clustering design}:\space
This design is possible only if $d|n$.
Let us partition the nodes into $n/d$ sets, each of which we call a \emph{cluster}.
Let us then assign each object to a cluster such that each cluster consists of exactly $d$ objects.
Every object is stored across all the nodes within its assigned cluster.
The resulting storage has an allocation graph composed of $n/d$ disjoint $d$-regular complete bipartite graphs.

For instance, the 3-choice allocation for nine objects $a, \dots, i$ with clustering design would look like
\begin{equation*}
  \icol{a\\b\\c} ~~\icol{a\\b\\c} ~~\icol{a\\b\\c} ~~\icol{d\\e\\f} ~~\icol{d\\e\\f} ~~\icol{d\\e\\f} ~~\icol{g\\h\\i} ~~\icol{g\\h\\i} ~~\icol{g\\h\\i}.
\end{equation*}

\vspace{1ex}
\noindent
\textbf{Cyclic design}:\space
In this design, we follow a \emph{cyclic} construction.
We start by assigning the original object copies to the nodes using an arbitrary bijection $f_0$.
We assign the $i$th replicas of the objects for $i = 1, \dots, d-1$ by using bijection $f_i$, where $f_i$ is obtained by applying circular shift on $f_0$ repeatedly $i$ times. Note that shifting is applied in the same direction in all steps.
In other words, we pick $f_i$ such that $f_{i+1}(o) = f_i(o) + 1 \mod n$ for $i = 0, \dots, d-1$ and every $o$.
For instance, 3-choice allocation for 7 objects $a, \dots, g$ with cyclic design would look like
\begin{equation*}
  \icol{a\\g\\f} ~~\icol{b\\a\\g} ~~\icol{c\\b\\a} ~~\icol{d\\c\\b} ~~\icol{e\\d\\c} ~~\icol{f\\e\\d} ~~\icol{g\\f\\e}.
\end{equation*}

Notably, a form of cyclic design is used in Cassandra \cite{Cassandra:LakshmanM10} and other similar storage systems \cite{Scylladb:Suneja19}.

t is not easy to get a general handle on the service choice overlaps and control them. 
The clustering design represents one of the extreme ways of distributing the service choice overlaps. With clustering design, objects within the same cluster fully overlap at all their service choices, while objects in different clusters do not. Cyclic design moves towards distributing the overlaps more evenly over different subsets of objects.
In order to state the distribution of overlaps more clearly, let us define the distance between two objects $o_i$ and $o_j$ as $|j - i|$. Then, with clustering or cyclic design, if two objects are apart with a distance greater than or equal to $d$, they do not overlap in their service choices.

Clustering and cyclic designs have been generalized in \cite{LoadBalancing:AktasFS21} to a class of allocations, namely \emph{$r$-gap design}, in which a single parameter loosely controls the service choice overlaps.
Constructions with $r$-gap design decouple service choices for objects that are $r$-apart.
\begin{definition}[\cite{LoadBalancing:AktasFS21}]
  A storage allocation is an \underline{$r$-gap design} if $\card{C_i \cap C_j} = 0$ for $j > i$ and $\min\{j-i, n-(j-i)\} > r$.
\label{def:rgap_design}
\end{definition}

Clustering and cyclic designs, or their generalization $r$-gap design, decouple service choices $C_i$ apart at the cost of enlarging the overlaps between those close to each other. We next discuss a different design that distributes the service choice overlaps evenly.

\vspace{1ex}
\noindent
\textbf{Block design}:\space
A $(d, v)$ block design is a class of equal-size subsets of $\mathcal{X}$ (the set of stored objects), called blocks (storage nodes), such that every point in $\mathcal{X}$ appears in exactly $d$ blocks (service choices), and \emph{every} pair of distinct points is contained in exactly $v$ blocks \cite{CombinatorialDesigns:Stinson07}.

We here consider the symmetric block designs with $v = 1$; that is, the number of objects and nodes are equal ($k = n$) and $|C_i \cap C_j| = 1$ for every $j \neq i$.
Note that $v = 1$ represents the setting with minimal overlap (maximum span) between the service choices $C_i$.
Using the fact that every pair of service choices overlaps at a single node, we obtain the cumulative overlap between service choice pairs (see Def.~\ref{def:cum_overlap}) as
\[ 
    \mathrm{CumOverlap}_2 = {n \choose 2}.
\]
This value must be equal to the value given by the expression in \eqref{eq:cum_overlap_in_balanced_d_choice}.
Thus, the block designs we consider here are possible only if $n = d^2 - d + 1$.
For instance, a 3-choice allocation with a block design is given as
\begin{equation}
  \icol{a\\b\\c} ~~\icol{a\\f\\g} ~~\icol{a\\d\\e} ~~\icol{b\\d\\f} ~~\icol{b\\e\\g} ~~\icol{c\\d\\g} ~~\icol{c\\e\\f}
\label{eq:eq_3choice_block_design}
\end{equation}

When the system parameters do not allow for constructing block design exactly, it is possible to obtain an \emph{approximate} one via a randomized construction. We discuss such a randomized construction in Appendix~\ref{subsec:approx_block_design_w_random_construction}.

\vspace{1ex}
\noindent
\textbf{Random design}:\space
In this design we follow a \emph{random} construction.
We start by assigning the original object copies to the nodes according to an arbitrary bijection $f_0$.
We assign the $i$th replicas of the objects for $i = 1, \dots, d-1$ by using bijection $f_i$, where $f_i$ implements a random permutation such that no object gets assigned to the same node more than once.
Random design may not yield a balanced allocation due to random assignment of objects to nodes. 
For instance, 3-choice allocation for 7 objects $a, \dots, g$ with random design may look like
\begin{equation*}
  \icol{a\\b\\f\\g} ~~\icol{b\\c} ~~\icol{c\\a\\g} ~~\icol{d\\b} ~~\icol{e\\a} ~~\icol{f\\c\\d\\e} ~~\icol{g\\d\\e\\f}.
\end{equation*}


\subsection{Robustness of A Given Storage Allocation}
\label{subsec:calculating_P_for_given_storage_allocation}

We discussed above the importance of service choice overlaps on the system's robustness $\mathcal{P}$ and presented four storage designs that implement different overlap distribution characteristics.
As discussed in Sec.~\ref{subsec:cap_region_via_service_choice_spans}, $\mathcal{P}$ depends on the service choice spans (or overlaps equivalently) for all subsets of objects, making it impossible to derive a single ``formula'' to calculate $\mathcal{P}$ for any given storage allocation. Analysis of $\mathcal{P}$ largely depends on the storage design. As shown in the following sections, we cannot reuse the techniques for one storage design while analyzing $\mathcal{P}$ for a different design.
This subsection presents an approach to derive bounds on $\mathcal{P}$ for any given $d$-choice storage allocation.

The main difficulty in analyzing $\mathcal{P}$ is that we cannot decouple object demands from each other while determining whether the system can serve them. For a given object-$i$, how we split and assign its demand $\rho_i$ depends on the demand assignment for the other objects overlapping with object-$i$ in their service choices.
This dependence is captured in \eqref{eq:modified_cap_region_w_service_choice_spans} by the set of conditions we should check to see if sufficient cumulative capacity is available to meet the demand for all subsets of objects. 
For brevity, we refer to these as \emph{demand-vs-capacity conditions}.

Iterating over the object demands $(\rho_1, \dots, \rho_n)$ and checking the relevant demand-vs-capacity conditions helps us understand the dependence that renders the analysis of $\mathcal{P}$ difficult.
In the beginning, we only consider $\rho_1$, and there is only one condition we need to check: $\rho_1 \leq |C_1|$.
At step $2$, include $\rho_2$, after which we need to check the following conditions: $\rho_2 \leq |C_2|$ and $\rho_1 + \rho_2 \leq |C_1 \cup C_2|$.
At step $3$, include $\rho_3$, after which we need to check the following conditions: $\rho_3 \leq |C_3|$, $\rho_1 + \rho_3 \leq |C_1 \cup C_3|$, $\rho_2 + \rho_3 \leq |C_2 \cup C_3|$ and $\rho_1 + \rho_2 + \rho_3 \leq |C_1 \cup C_2 \cup C_3|$.
We must continue these steps until we include all $\rho_i$ and check all the demand-vs-capacity conditions.
The dependence that renders the analysis difficult can be seen clearly at step $3$. Conditions in which $\rho_3$ is included depend on the values of the previous demands $\rho_1$ and $\rho_2$, as well as the previous service choices $C_1$ and $C_2$.
This dependence makes it impossible to decouple the conditions, calculate the probability of meeting them, and then combine them to find the value for $\mathcal{P}$.

Although an exact analysis of $\mathcal{P}$ is formidable for an arbitrary storage allocation, it is possible to find bounds on $\mathcal{P}$ by focusing on a subset of the demand-capacity conditions.
It is evident that focusing on only a subset of the demand-capacity conditions given in \eqref{eq:modified_cap_region_w_service_choice_spans} would give us a larger (modified) capacity region $\mathcal{C}_m$ than the actual one.
Then, calculating $\mathcal{P}$ concerning a larger $\mathcal{C}_m$ using \eqref{eq:P_integration} would yield an upper bound on its actual value.
We present such an upper bound in Theorem~\ref{thm:P_upper_bound_for_any_storage_allocation}.
To state the bound, we need the following definition.

\begin{definition}
  We denote the \underline{span} of an arbitrary set of $t$ objects in a storage allocation with $\mathrm{span}_t$.
  Its distribution is given as
  \begin{equation*}
    \Prob{\mathrm{span}_t = s} = \frac{\card{\Bigl\{(i_1, ..., i_t) \;\Bigl|\; \mathrm{span}(o_{i_1}, ..., o_{i_t}) = s\Bigr\}}}{{n \choose t}}.
  \end{equation*}
\label{def:Pr_span}
\end{definition}

\begin{theorem}
  For a given storage allocation,
  \begin{equation}
    \mathcal{P} < \min_{1 \leq t \leq n} \mathcal{P}_t,
  \label{eq:P_upper_bound_for_any_storage_allocation}
  \end{equation}
  in which the right-hand side of the inequality is given by
  \begin{equation}
    \mathcal{P}_t = \Exp{F_t(m \cdot \mathrm{span}_t)},
  \label{eq:P_t}
  \end{equation}
  where the expectation is with respect to random variable $\mathrm{span}_t$ with a distribution given in Def.~\ref{def:Pr_span}.
\label{thm:P_upper_bound_for_any_storage_allocation}
\end{theorem}
\begin{proof}
  See Appendix~\ref{subsec:proof_thm_P_upper_bound_for_any_storage_allocation}.
\end{proof}

Distribution $F_t(x)$ given in \eqref{eq:P_t} for the sum $\rho_1 + \dots + \rho_t$ can be written in closed form for certain demand distributions. For instance, the sum would have a Gamma distribution if $\rho_i \sim \mathrm{Exp}$, or a Binomial distribution if $\rho_i \sim \mathrm{Bernoulli}$.
It is also worth noting that the upper bound in \eqref{eq:P_upper_bound_for_any_storage_allocation} can be improved by including the conditions $\rho_i \leq d$ and updating \eqref{eq:P_t} as
\[
    \mathcal{P}_t = \Exp{\Prob{\rho_1 \leq d, \dots, \rho_n \leq d, \rho_1 + \dots + \rho_t \leq m \cdot \mathrm{span}_t}}.
\]
This, however, makes it harder to compute the upper bound.

A comparison between the upper bound in \eqref{eq:P_upper_bound_for_any_storage_allocation} and the simulated values is given in Fig.~\ref{fig:P_sim_vs_upper_bound_for_given_allocation}. As shown in these plots and others that we omitted here, the upper bound gives almost the same values for clustering, cyclic, and random designs, as expressed by the curves placed on each other. This indicates that the upper bound given in \eqref{eq:P_upper_bound_for_any_storage_allocation} does not capture the impact of different ways of distributing the service choice overlaps. For the demand distributions $\mathrm{Exp}$ and $\mathrm{Pareto}$ used while generating the plots, cyclic design performs better than clustering design, and random design achieves the highest $\mathcal{P}$. This is not always the case for all demand distributions, as explained in the following.
The upper bound is reasonably tight on $\mathcal{P}$ for random design (the best-performing storage design), and tight for large values of $d$ (see the plots for $d = 10$).

\begin{figure*}[t]
  \centering
  \begin{subfigure}
    \centering
    \includegraphics[width=.45\textwidth]{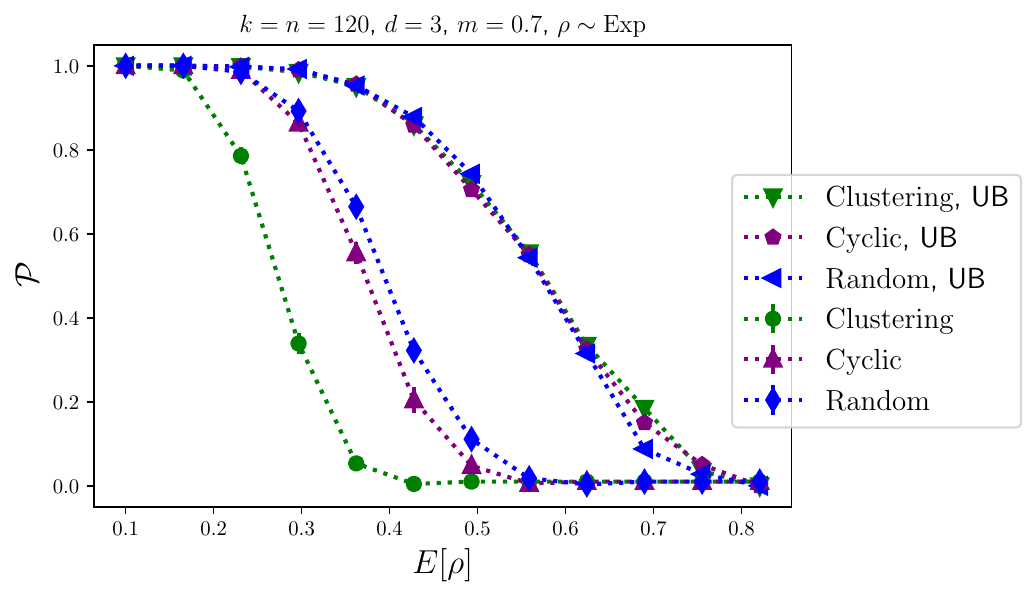}
  \end{subfigure}
  \begin{subfigure}
    \centering
    \includegraphics[width=.45\textwidth]{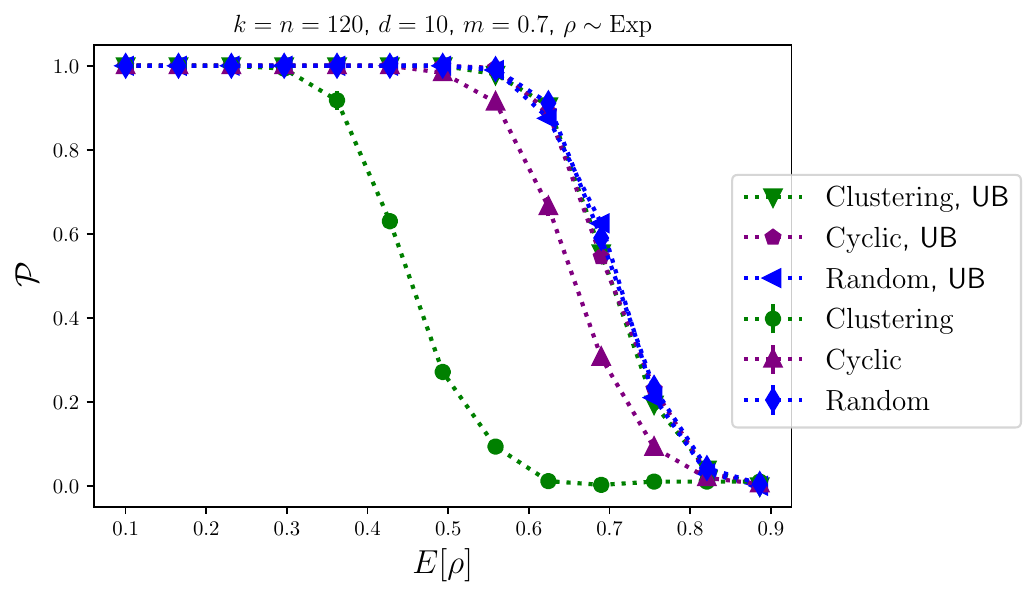}
  \end{subfigure}
  \begin{subfigure}
    \centering
    \includegraphics[width=.45\textwidth]{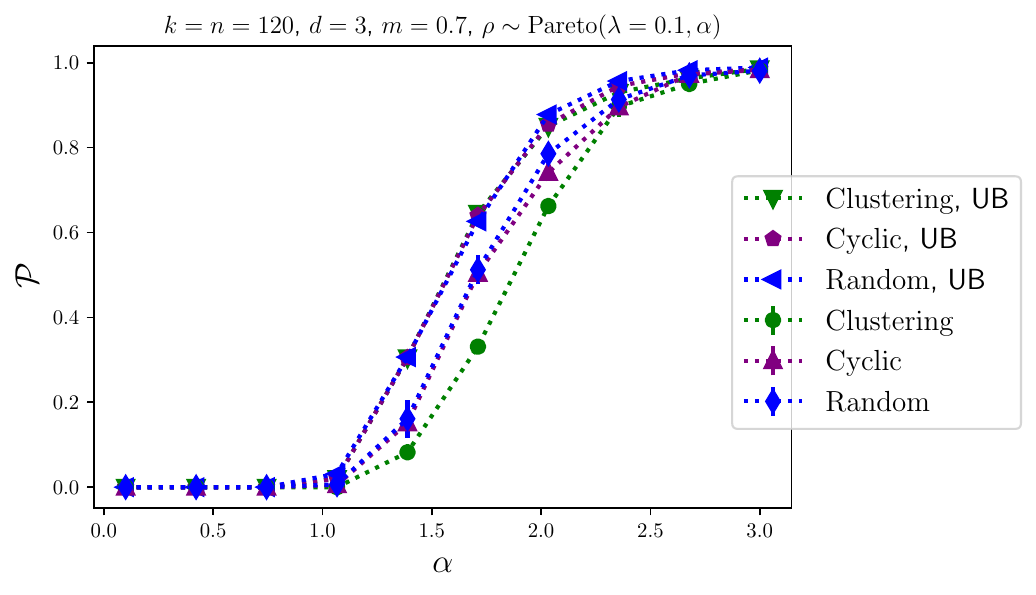}
  \end{subfigure}
  \begin{subfigure}
    \centering
    \includegraphics[width=.45\textwidth]{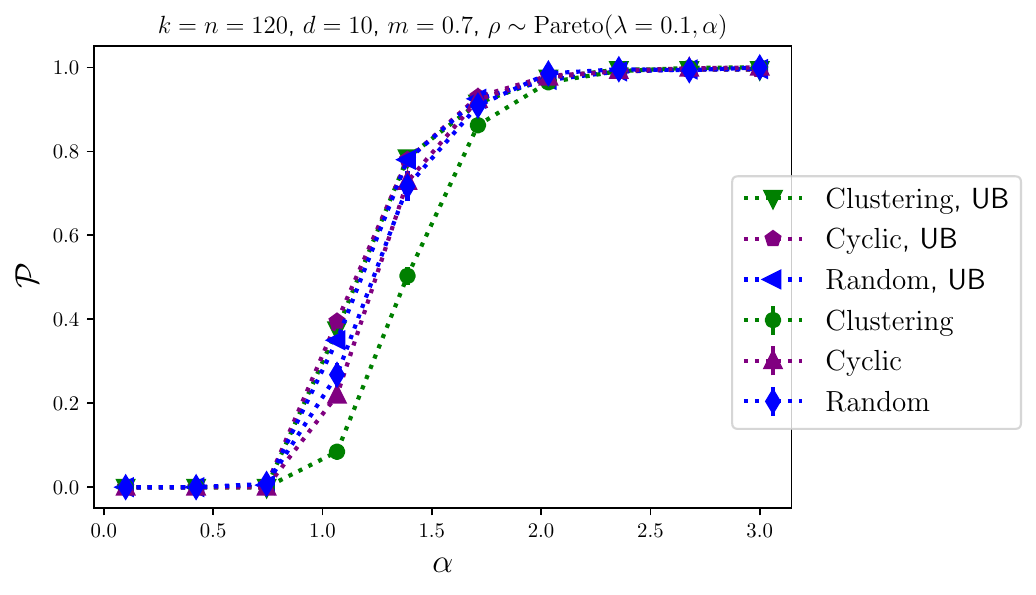}
  \end{subfigure}
  \caption{
    Simulated $\mathcal{P}$ values vs. the upper bound presented in Theorem~\ref{thm:P_upper_bound_for_any_storage_allocation} for $d$-choice allocation with different designs.
    Each plot is for a different $d$ ($3$ or $10$) and for a different demand distribution ($\mathrm{Exp}$ or $\mathrm{Pareto}$). 
    Suffix ``UB'' in the legend refers to the upper bound.
    Note that Pareto distribution becomes (stochastically) smaller as $\alpha$ becomes larger, so does $\mathcal{P}$.
  }
\label{fig:P_sim_vs_upper_bound_for_given_allocation}
\end{figure*}

\vspace{1ex}
\noindent
\textbf{Impact of Service Choice Overlaps}:\space
We elaborate on different strategies for distributing service choice overlaps regarding their impact on $\mathcal{P}$.
As explained in Sec.~\ref{subsec:balanced_d_choice}, cumulative service choice overlaps are the same across all balanced $d$-choice allocations. However, different designs distribute the overlaps according to different strategies.

Fig.~\ref{fig:P_sim_for_design_comparison} shows the simulated values of $\mathcal{P}$ for $3$-choice allocation constructed with clustering, cyclic, block and random design. It assumes that the demand distribution is  $\lambda \times \mathrm{Bernoulli}(p)$ for fixed $\lambda$ and varying $p$. An object is either actively requested with a demand of $\lambda$ or not requested at all. Note also that $\mathcal{P}$ is calculated for maximal load $m = 1$; it quantifies the probability that the system can serve the offered demand under stability.
When $D = 2$, $\mathcal{P}$ for different storage designs is ordered as 
\begin{equation}
    \mathcal{P}_{\mathrm{block}} > \mathcal{P}_{\mathrm{random}} > \mathcal{P}_{\mathrm{cyclic}} > \mathcal{P}_{\mathrm{clustering}}.
\label{eq:P_order_for_storage_designs}
\end{equation}
We present the plots for $\mathrm{Bernoulli}$ demand distribution only, but we have found that the order in \eqref{eq:P_order_for_storage_designs} holds for other demand distributions as well. Examples with other demand distributions can be found in Fig.~\ref{fig:P_sim_vs_upper_bound_for_given_allocation}. This observation suggests that \emph{designing many but consistently small service choice overlaps performs better than fewer but occasionally larger overlaps}.
Notice that random design performs close to block design, which is an important observation for practical systems. It suggests that we cannot substantially improve the system's robustness (in terms of $\mathcal{P}$) by going beyond random design and distributing service choice overlaps evenly with block design.

\begin{figure*}[t]
  \centering
  \begin{subfigure}
    \centering
    \includegraphics[width=.45\textwidth]{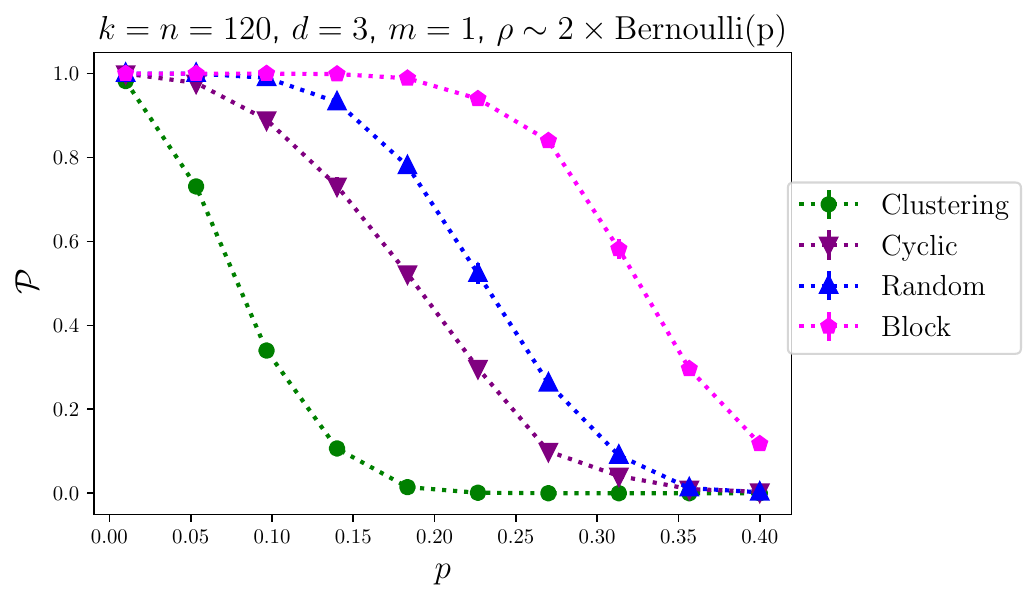}
  \end{subfigure}
  \begin{subfigure}
    \centering
    \includegraphics[width=.45\textwidth]{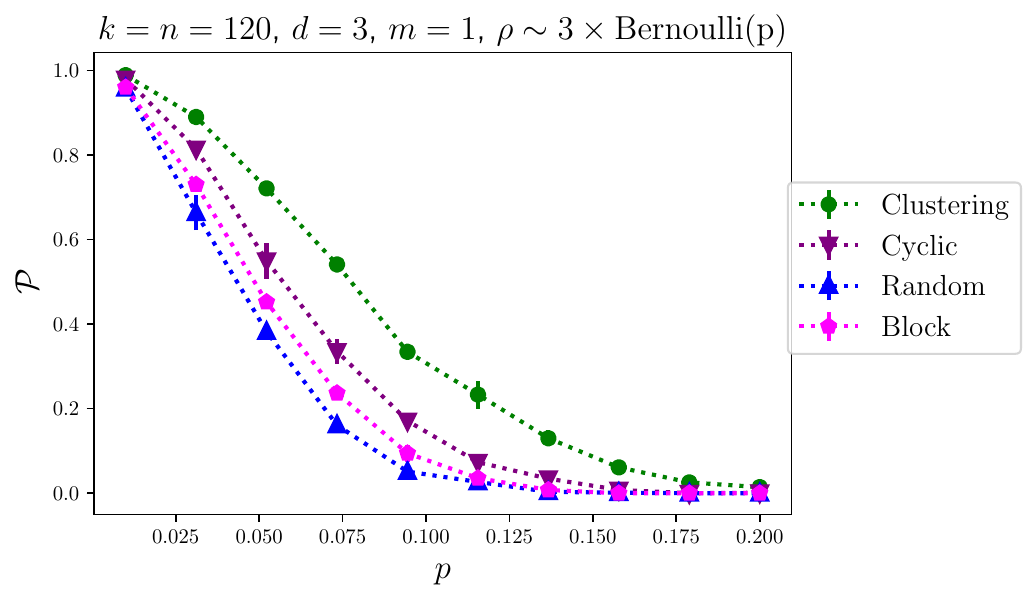}
  \end{subfigure}
  \caption{Simulated $\mathcal{P}$ for different storage designs with varying replication factor $d$.}
\label{fig:P_sim_for_design_comparison}
\end{figure*}

As also shown in Fig.~\ref{fig:P_sim_for_design_comparison}, somewhat interestingly, the order of $\mathcal{P}$ given in \eqref{eq:P_order_for_storage_designs} over the storage designs is reversed when the demand of active objects $\lambda$ is set to $d$. Note that we observe this when $\lambda = d$ and when $\lambda$ is sufficiently close to $d$.
When the demand is ``skewed beyond a level'', the impact of service choice overlaps becomes the opposite of what it was: \emph{designing few but occasionally large overlaps performs better than more but consistently smaller overlaps}.

Storage designs that favor many but consistently small overlaps, such as block design, use the following rationale: demand for most objects will likely be around an average. In contrast, objects with very large or small demands will be rare. If this presumption is true, then it makes sense to implement uniformly medium-size service choices over all subsets of objects. This is better than allocating large service choice spans for a few subsets of objects at the expense of allocating small spans for the rest, as the expected demand will not need a large span.
This situation is represented by the case with $\lambda = 2$ discussed above. Hence we observe the expected order in $\mathcal{P}$ as stated in \eqref{eq:P_order_for_storage_designs}.

On the other hand, storage designs that favor fewer but occasionally larger overlaps, such as clustering design, use the following rationale: object clusters that overlap at many nodes are not likely to be jointly popular, and popular objects are likely to fall into different clusters. This would support demand vectors in which a few ``active'' objects have a large demand (close to the maximum $d$ that can be supported for a single object) while the rest have a small or negligible demand.
In this case, depending on the degree of difference between the demands of active and non-active objects, allocating large service choice spans to a limited collection of subsets of objects while allocating small spans to others (i.e., fewer but occasionally larger overlaps) might make sense.

In order to understand the performance metric $\mathcal{P}$ with rigor for a given replication factor $d$ and storage design, when $\lambda < d$, so far we have only the upper bound presented in Theorem~\ref{thm:P_upper_bound_for_any_storage_allocation}. We will separately present more accurate expressions and bounds for clustering, cyclic, block, and random design in the following sections.
However, when $\lambda = d$, we can find exact expressions of $\mathcal{P}$ for each storage design, as presented in the following Theorem.
This is because when $\lambda = d$, once an object has non-zero demand (active), it needs to use up all the available capacity ($d \cdot m$) at all the nodes that host the object. As each active object books all its service choices, the system can meet the maximum demand requirement only if the set of active objects does not overlap in their service choices. This allows us to decouple the assignment of individual object demands to nodes while analyzing $\mathcal{P}$.
\begin{theorem}
  When the demand distribution is $m d \times \mathrm{Bernoulli}(p)$, for $d$-choice allocation with
  
  \noindent
  (1) Clustering (assuming $d | n$), cyclic or block design:
  \begin{equation}
    \mathcal{P} = \mathbb{E}_A\left[\I{n - A \cdot c + 1} \cdot \prod_{i=1}^{A - 1} \frac{n - i \cdot c}{n - i}\right],
  \label{eq:P_demand_is_bernoulli_lambda_eq_d_for_clustering_cyclic_block}
  \end{equation}
  where (a) $c = d$ for clustering, (b) $c = 2d - 1$ for cyclic, (c) $c = d^2 - d + 1$ for block design, and $A \sim \mathrm{Binomial}(n, p)$.

  \noindent
  (2) Random design:
  \begin{equation}
    \mathcal{P} = \mathbb{E}_A\left[\I{n - A \cdot d + 1} \cdot \prod_{i=1}^{A - 1} \frac{{n - i \cdot d \choose d}}{{n \choose d}}\right],
  \label{eq:P_demand_is_bernoulli_lambda_eq_d_for_random}
  \end{equation}
  where $A \sim \mathrm{Binomial}(n, p)$.
\label{thm:P_demand_is_bernoulli_lambda_eq_d}
\end{theorem}
\begin{proof}
  See Appendix~\ref{subsec:proof_thm_P_demand_is_bernoulli_lambda_eq_d}.
\end{proof}

From \eqref{eq:P_demand_is_bernoulli_lambda_eq_d_for_clustering_cyclic_block}, we can clearly see why $\mathcal{P}_{\mathrm{clustering}} > \mathcal{P}_{\mathrm{cyclic}} > \mathcal{P}_{\mathrm{block}}$ for $D = d$. Each time an object is selected to be active with demand $d$, any object that overlaps with it at any of its service choices must not be active for the system to satisfy the maximal demand requirement. The number of objects overlapping with a given object is given by the variable $c$ in this case, which is largest for block design and smallest for clustering design. Hence the order of $\mathcal{P}$.
The expression given in \eqref{eq:P_demand_is_bernoulli_lambda_eq_d_for_clustering_cyclic_block} can be simplified for clustering design as follows
\begin{equation}
    \mathcal{P} = \mathbb{E}_A\left[{{n / d \choose A}}\bigg/{\binom{n}{A}}\right],
\label{eq:P_demand_is_bernoulli_lambda_eq_d_for_clustering}
\end{equation}

Each storage design requires a fundamentally different approach to analyze $\mathcal{P}$ unless the demand distribution allows for decoupling the assignment of individual object demands -- such as the case presented in Theorem~\ref{thm:P_demand_is_bernoulli_lambda_eq_d}.
The upper bound presented in Theorem~\ref{thm:P_upper_bound_for_any_storage_allocation} applies for any storage design, but it is not sufficiently tight for all storage designs and system parameters of interest.
In the following subsections, we will show that we can find tighter bounds for $\mathcal{P}$, and even exact expressions sometimes, by focusing on individual storage designs. This enables us to use the specialized mathematical tools that fit the characteristics of the service choice overlaps implemented by the storage design under consideration.

\subsection{Clustering Design}
\label{subsec:clustering_design}
We refer the reader to Sec.~\ref{subsec:balanced_d_choice} for the description of the construction of $d$-choice allocation with clustering design. Recall that this design is possible only if $d | n$.

Thanks to its tractable structure, we can find an exact expression of $\mathcal{P}$ for systems with clustering design.

\begin{lemma}
  In $d$-choice allocation with clustering design
  \begin{equation}
    \mathcal{P} = F_d(m \cdot d)^{n / d}.
  \label{eq:P_clustering}
  \end{equation}
\label{lm:P_clustering}
\end{lemma}
\begin{proof}
  See Appendix~\ref{subsec:proof_lm_P_clustering}.
\end{proof}

Expression \eqref{eq:P_clustering} enables us to evaluate $\mathcal{P}$ for varying values of $d$. Using \eqref{eq:P_clustering}, it is straightforward to conclude that $\mathcal{P}$ grows faster than polynomial $x^{1/d}$, i.e., $\mathcal{P}_{1} < \mathcal{P}_{d}^{d}$ where $\mathcal{P}_{d}$ denotes $\mathcal{P}$ for $d \geq 1$. Tighter results can be derived for certain distributions of $\rho_i$. For instance, $\mathcal{P}$ grows faster than $x^{1/d^2}$ for $\rho_i \sim \mathrm{Exp}$.
This is demonstrated by the $\mathcal{P}$ vs. $d$ curves plotted in Fig.~\ref{fig:P_clustering_vs_d} for $\rho_i \sim \mathrm{Exp}$.

\begin{figure}[hbt]
  \centering
    \includegraphics[width=.4\textwidth]{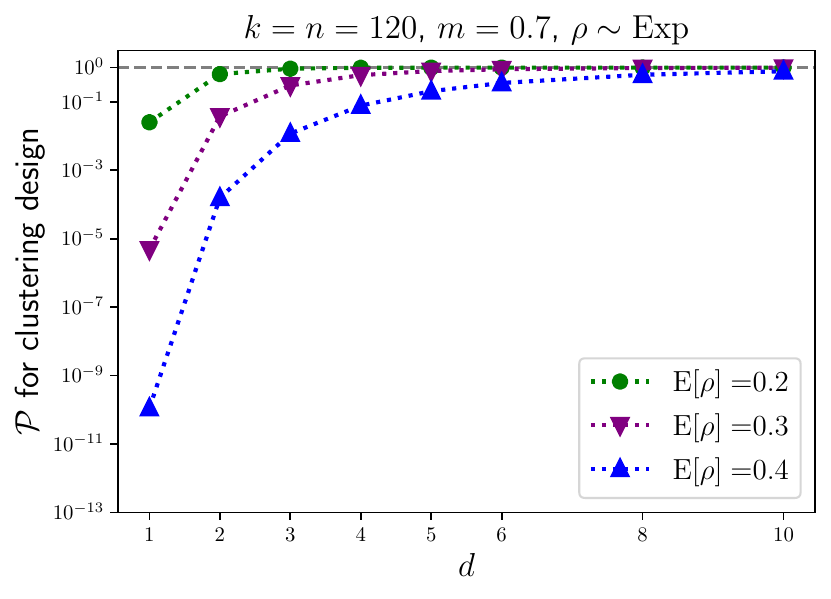}
  \caption{$\mathcal{P}$ vs replication factor $d$ for clustering design.}
\label{fig:P_clustering_vs_d}
\end{figure}

Expression \eqref{eq:P_clustering} also enables us to find useful bounds for $\mathcal{P}$.
For instance, using Chernoff bound, we find the following
\begin{equation}
    \mathcal{P} \geq \sup_{s > 0} \left[1 - \exp\left(-d \left(m \cdot s- \ln(\Phi_{\rho}(s))\right)\right) \right]^{n / d},
\label{eq:P_clustering_upper_bound}
\end{equation}
where $\Phi_{\rho}(s)$ denotes the moment generating function for $\rho_i$.
For instance, if $\rho_i \sim \mathrm{Exp}(\mu)$, \eqref{eq:P_clustering_upper_bound} would be given as
\[
    \mathcal{P} \geq \left[1 - \exp\left(-d \left(m\mu - \ln(m\mu) - 1\right)\right) \right]^{n / d}.
\]

Using \eqref{eq:P_clustering} and concentration inequalities, we can find not only a lower bound for $\mathcal{P}$ but also an upper bound for more general demand distributions.
\begin{theorem}
  Suppose that object demands $\rho_i$ are sub-gaussian in the sense that there exists constants $c, C > 0$ such that
  \[
    \exp(c t^2) \leq \phi_{\rho_i}(t) \leq \exp(C t^2), \qquad \forall t > 0.
  \]
  Then, in a $d$-choice storage allocation constructed with clustering design, there exists constants $\alpha, \beta, \gamma > 0$ such that
  \begin{equation}
  \begin{split}
    \mathcal{P} &\geq \left(1 - \gamma \cdot \exp\left(-d \cdot \beta(m - \mu)^2\right)\right)^{n/d} \\
    \mathcal{P} &\leq \left(1 - \exp\left(-d \cdot \alpha(m - \mu)^2\right)\right)^{n/d},
  \end{split}
  \label{eq:P_bounds_for_clustering}
  \end{equation}
  where $\mu = \Exp{\rho_i}$ and $m > \mu$.
\label{thm:P_bounds_for_clustering}
\end{theorem}
\begin{proof}
  See Appendix~\ref{subsec:proof_thm_P_bounds_for_clustering}.
\end{proof}

\begin{remark}
  Even though we stated the bounds given in \eqref{eq:P_bounds_for_clustering} for sub-gaussian demand distributions, we can get bounds in the same form for demand distributions with heavier tail such as sub-exponential distributions.
  These bounds show that $\mathcal{P}$ scales with $d$, $n$ and $m$ as
  \begin{equation}
    \mathcal{P} \;\myapprox\; \left(1 - \exp\left(-d \cdot K(m - \mu)^2\right)\right)^{n/d}
  \label{eq:P_scaling_with_d_n_m}
  \end{equation}
  for some constant $K > 0$.
  Most importantly, this quantifies the performance gain we can achieve in $\mathcal{P}$ by increasing the replication factor $d$.

  First, \eqref{eq:P_bounds_for_clustering} validates our intuition that incrementing $d$ yields diminishing return in increasing $\mathcal{P}$.
  This can be seen visually in the example plots shown in Fig.~\ref{fig:P_clustering_vs_d}.
  When $d \cdot K(m - \mu)^2 < 1$, we can use Taylor expansion to approximate $\mathcal{P}$ as
  \begin{equation}
    \mathcal{P} \approx \left(d \cdot K(m - \mu)^2\right)^{n/d}.
  \label{eq:P_approx_small_d_for_clustering}
  \end{equation}

  Second, \eqref{eq:P_bounds_for_clustering} shows that setting $d$ close to $\log(n)$ helps us increase $\mathcal{P}$ substantially.
  This can also be seen in Fig.~\ref{fig:P_clustering_vs_d}.
  Recall from \eqref{eq:P_scaling_with_d_n_m} that, for $d = 1$ and $M = K(m - \mu)^2$, we have
  \[
    \mathcal{P} \;\myapprox\; \left(1 - \exp\left(-M\right)\right)^{n}
  \]
  Suppose that $n = 1000$ and $M = 5$. Setting $d = 1$ then gives $\mathcal{P} \approx 0.001$.
  By setting $d = \log(n)$, we raise $\mathcal{P}$ up to 
  \[
    \left(1 - n^{-M}\right)^{n/\log(n)}.
  \]
  When $d = 8$, which is the smallest $d$ such that $d \geq \log(n)$ and $d | n$, we get $\mathcal{P} \approx 1$. Thus, in this case, it does not make sense to increment the replication factor $d$ beyond $\log(n)$.
  \closeremark
\end{remark}

Using \eqref{eq:P_scaling_with_d_n_m}, we can also study the behavior of $\mathcal{P}$ in the limit as the replication factor $d$ scales with the system scale $n$.
\begin{corollary}
    In $d$-choice allocation with clustering design
    \begin{equation}
    \mathcal{P} \to
    \begin{cases} 
      1 & d = \Omega(\log(n)) \\
      0 & \text{otherwise}
    \end{cases}
    \quad \text{as } n \to \infty.
    \label{eq:P_limit_for_clustering}
    \end{equation}
\label{cor:P_for_clustering_as_d_to_infty}
\end{corollary}
\begin{proof}
  See Appendix~\ref{subsec:proof_cor_P_for_clustering_as_d_to_infty}.
\end{proof}

\begin{remark}
  The limit in \eqref{eq:P_limit_for_clustering} shows that, in systems with clustering design, replication factor $d$ needs to grow at least as fast as $\log(n)$ to meet the maximal load requirement as the system scale $n$ gets larger.
  In the next subsection, we show that the same observation holds for systems with cyclic design.
  \closeremark
\end{remark}

\subsection{Scan Statistics Interlude}
\label{subsec:scan_stats}
We here define and discuss \emph{scan statistics}.
It is a mathematical object that is instrumental while deriving our results on $\mathcal{P}$ for cyclic, block and random designs.
We make frequent use of scan statistic defined over the sequence of object demands. 
We also present an upper bound on $\mathcal{P}$ for $d$-choice allocation with any storage design in terms of scan statistic.

\begin{definition}
  The \underline{$s$-scan} of $n$ i.i.d. random variables $X_i$ is defined as the sequence of random variables
  \begin{equation*}
      S_{s, i} = \sum_{j=i}^{i + s - 1} X_i \quad \text{for } i \in [1, n - s + 1].
  \end{equation*}
  When $i \in [1, n]$ and $X_i = X_{i \mod n}$, we refer to the scanning sequence as the \underline{circular $s$-scan} and denote it with $S_{s, i}^{(c)}$.

  The \underline{scan statistic} is defined as the maximum scan as $S_s = \max_{1 \leq i \leq n - s + 1} S_{s, i}$.
  Similarly the \underline{circular scan statistic} is defined as $S_s^{(c)} = \max_{1 \leq i \leq n} S_{s, i}^{(c)}$.
\label{def:scan_statistic}
\end{definition}

The circular scan statistic $S^{(c)}_s$ converges to its non-circular counterpart $S_s$ almost surely as the number of random samples $n \to \infty$ (see Appendix~\ref{subsec:circular_scan} for the proof). We make use of this observation while we analyze $\mathcal{P}$ for cyclic design in this section.
In the remainder of the paper, the statistics $S_s$ and $S_s^{(c)}$ are defined over the sequence of object demands $\rho_i$.

Main utility of the scan statistic is that we can find bounds on $\mathcal{P}$ in terms of $S_s^{(c)}$ as shown in this section. This allows applying the results available on scan statistic to derive insight on system performance.
Scan statistics has many applications in different fields. We here present a new application of scan statistics in the context of distributed storage performance.
For a thorough exposition to scan statistics literature, we refer the reader to \cite{ApproxForScanStatDist:Naus82, LimitLawsOfMovingSum:DeheuvelsD87, PoissonApproxForScanStat:DemboK92, ScanStats:GlazNW01, ExtremeValueMethodsWithAppsToFinance:Novak11}.
We here discuss only some of the results on scan statistics which helps us analyze $\mathcal{P}$.

Using the idea introduced in Sec.~\ref{subsec:calculating_P_for_given_storage_allocation}, we can find an upper bound on $\mathcal{P}$ for any $d$-choice allocation by focusing on only the demand-capacity conditions defined on the consecutive objects of a fixed length. We can express such an upper bound in terms of scan statistic as given in the following.
\begin{lemma}
  In $d$-choice storage allocation, we have
  \begin{equation}
      \mathcal{P} \leq \min_{1 \leq s \leq n}\Prob{S^{(c)}_s \leq s \cdot md}.
  \label{eq:P_upper_bound_for_any_design}
  \end{equation}
\label{lm:P_upper_bound_for_any_design}
\end{lemma}
\begin{proof}
  See Appendix~\ref{subsec:proof_lm_P_upper_bound_for_any_design}.
\end{proof}

In Corollary~\ref{cor:P_upper_bound_for_any_design_insightful} below, we present several results on $\mathcal{P}$ that we build on \eqref{eq:P_upper_bound_for_any_design}. We do not present the proofs for these statements as their proofs closely follow the statements and proofs presented for cyclic design in the next sub-section. 
We state them here to demonstrate what scan statistics offers to understand performance of $d$-choice allocation with any storage design.
For instance, \eqref{eq:P_convergence_for_any_design} given below shows that $d$ must scale at least as fast as $\log(n)^{1/3}$ as the system scale $n$ gets larger.
\begin{corollary}
  By substituting $s = d$ in \eqref{lm:P_upper_bound_for_any_design}, we obtain a more insightful version of the upper bound as
  \begin{equation}
    \mathcal{P} \leq \Prob{S_d^{(c)} \leq m \cdot d^2}.
  \label{eq:P_upper_bound_for_any_design_insightful}
  \end{equation}
  Using this, we find that in the limit $n \to \infty$
    \begin{equation}
      \mathcal{P} \leq 
      \exp\left(-w_{n, d} \; Q_d(m \cdot d^2)\right),
    \label{eq:P_asymptotic_upper_bound_for_any_design}
    \end{equation}
    where $w_{n, d} = n - d + 1$ and $Q_d(x) = 1 - F_d(x)$. 
    This implies the following in the limit $n \to \infty$
    \begin{equation}
    \begin{split}
    \mathcal{P} &\leq \exp\left(-w_{n, d} \cdot \exp\left(-d \cdot \alpha(md - \mu)^2\right)\right).
    \label{eq:P_asymptotic_upper_bound_for_any_design_insightful}
    \end{split}
    \end{equation}
    This finally implies that as $n \to \infty$
    \begin{equation}
        \mathcal{P} \to 0 \;\;\text{if}\;\; d = o\left(\log(n)^{1/3}\right).
    \label{eq:P_convergence_for_any_design}
    \end{equation}
\label{cor:P_upper_bound_for_any_design_insightful}
\end{corollary}

\subsection{Cyclic Design}
\label{subsec:cyclic_design}
We refer the reader to Sec.~\ref{subsec:balanced_d_choice} for the description of the construction of $d$-choice allocation with cyclic design, and its generalization $r$-gap design. As discussed previously, $r$-gap design controls the service choice overlaps/spans loosely by a single parameter $r$. Lemma~\ref{lm:on_rgap} below makes this statement concrete.
Throughout this section, object $o_i$ implicitly denotes $o_{i \mod n}$ throughout.

\begin{lemma}
  In $d$-choice allocation with $r$-gap design, for a set of consecutive objects $S = \{o_i, o_{i+1}, \dots, o_{i+x-1}\}$ where $i, x \in [1, n]$, we have $r \geq d-1$ and $x \leq \mathrm{span}(S) \leq x + 2r$.
\label{lm:on_rgap}
\end{lemma}

By definition (see Def.~\ref{def:rgap_design}), $r$-gap design decouples the service choices for objects that are at least $r$-apart (in terms of their indices).
Lemma~\ref{lm:on_rgap} shows that this translates into controlling the span for consecutive objects $o_i, o_{i+1}, \dots, o_{i+x-1}$ of any length $x$.
Then using the same idea introduced in Sec.~\ref{subsec:calculating_P_for_given_storage_allocation}, we can find bounds on $\mathcal{P}$ by only focusing on the demand-capacity conditions defined on the consecutive objects of fixed length.




\begin{lemma}
  In $d$-choice allocation with $r$-gap design, we have
  \begin{equation}
    \Prob{S^{(c)}_{r + 1} \leq m d} \leq \mathcal{P} \leq \Prob{S^{(c)}_s \leq m(s + 2r)}.
  \label{eq:P_bounds_for_rgap}
  \end{equation}
  for $s = 1, \dots, n-2r$.
\label{lm:P_bounds_for_rgap}
\end{lemma}
\begin{proof}
  See Appendix~\ref{subsec:proof_lm_P_bounds_for_rgap}.
\end{proof}

\begin{figure*}[t]
  \centering
  \begin{subfigure}
    \centering
    \includegraphics[width=.45\textwidth]{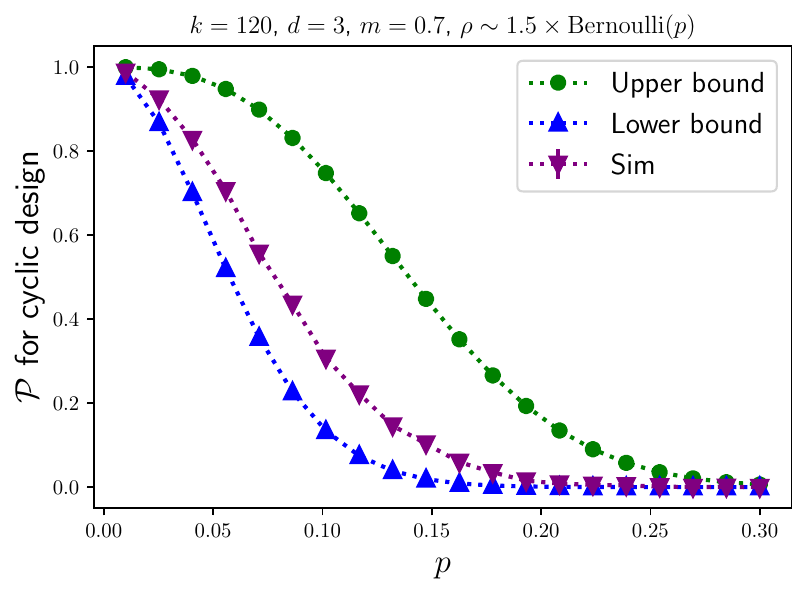}
  \end{subfigure}
  \begin{subfigure}
    \centering
    \includegraphics[width=.45\textwidth]{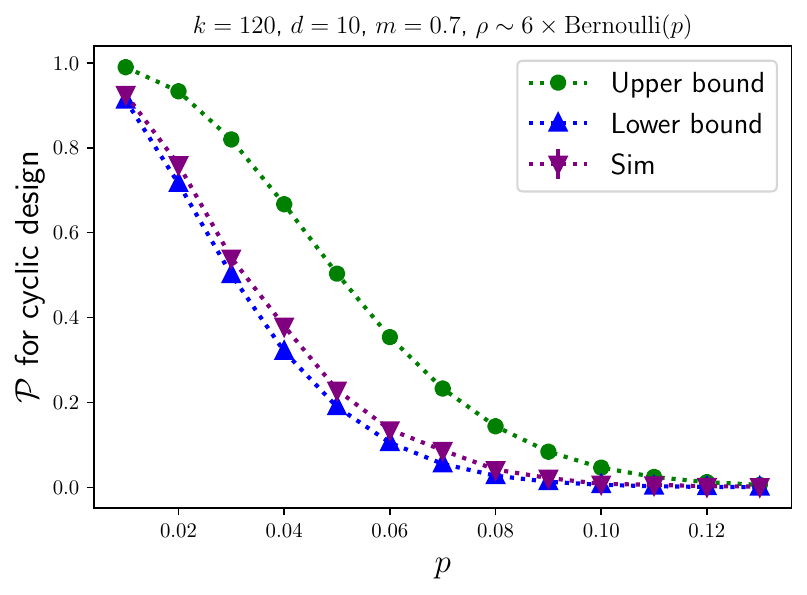}
  \end{subfigure}
  \caption{Simulated $\mathcal{P}$ values vs. the bounds presented in Lemma~\ref{lm:P_bounds_for_cyclic} for $d$-choice allocation with cyclic design.
  }
\label{fig:P_sim_vs_bounds_for_cyclic}
\end{figure*}

Notice that cyclic design (as well as clustering) is an $r$-gap design. Hence the bounds given in Lemma~\ref{lm:P_bounds_for_rgap} are valid for cyclic design. In addition, the well-defined structure of cyclic design allows us to refine these bounds as follows.
\begin{lemma}
  In $d$-choice allocation with cyclic design, we have
  \begin{equation}
  \begin{split}
    \mathcal{P} &\geq \Prob{S_d^{(c)} \leq m d}, \\
    \mathcal{P} &\leq \min_{1 \leq s \leq n - d + 1} \Prob{S^{(c)}_s \leq m(s + d - 1)}.
  \end{split}
  \label{eq:P_bounds_for_cyclic}
  \end{equation}
\label{lm:P_bounds_for_cyclic}
\end{lemma}
\begin{proof}
  See Appendix~\ref{subsec:proof_lm_P_bounds_for_cyclic}.
\end{proof}

Lemma~\ref{lm:P_bounds_for_cyclic} allows us to apply the results on the scan statistics to analyze and calculate $\mathcal{P}$ for cyclic design.
Fig.~\ref{fig:P_sim_vs_bounds_for_cyclic} shows the simulated values for $\mathcal{P}$ together with the values computed using \eqref{eq:P_bounds_for_cyclic}. The object demands are assumed to be distributed as $\lambda \times \mathrm{Bernoulli}(p)$, and the plots are given for different values of $d$ and $\lambda$.
We use the expression in Theorem~2, \cite{ApproxForScanStatDist:Naus82} presented for scan statistic distribution to compute the bounds in \eqref{eq:P_bounds_for_cyclic}.

We next find asymptotic bounds for $\mathcal{P}$ by leveraging the Poisson limit for scan statistic as presented in Theorem 2, \cite{PoissonApproxForScanStat:DemboK92}.
\begin{theorem}
  Suppose the object demands $\rho_i > 0$. Then, for $d$-choice allocation with cyclic design, in the limit $n \to \infty$
  \begin{equation}
  \begin{split}
    \mathcal{P} &\geq \exp\left(-w_{n, d} \cdot Q_d(m d)\right), \\
    \mathcal{P} &\leq \min_{1 \leq s \leq n - d + 1} \exp\left(-w_{n, s} \cdot Q_s(m(s + d - 1))\right).
  \end{split}
  \label{eq:P_asymptotic_bounds_for_cyclic}
  \end{equation}
  where $w_{n, u} = n - u + 1$ and $Q_u(x) = 1 - F_u(x), \; \forall u > 0$.
\label{thm:P_asymptotic_bounds_for_cyclic}
\end{theorem}
\begin{proof}
  See Appendix~\ref{subsec:proof_thm_P_asymptotic_bounds_for_cyclic}.
\end{proof}

We next present the bounds in \eqref{eq:P_asymptotic_bounds_for_cyclic} in a more tractable form.
\begin{corollary}
    In $d$-choice allocation with cyclic design, in the limit $n \to \infty$ we have
    \begin{equation}
      \exp\left(-w_{n, d} \; Q_d(md)\right) \leq 
      \mathcal{P} \leq 
      \exp\left(-w_{n, d} \; Q_d(2md)\right),
    \label{eq:P_asymptotic_bounds_for_cyclic_simpler}
    \end{equation}
    where $w_{n, d} = n - d + 1$ and $Q_d(x) = 1 - F_d(x)$.
\label{cor:P_asymptotic_bounds_for_cyclic_simpler}
\end{corollary}
\begin{proof}
    See Appendix~\ref{subsec:proof_cor_P_asymptotic_bounds_for_cyclic_simpler}.
\end{proof}

The bounds in \eqref{eq:P_asymptotic_bounds_for_cyclic_simpler} show how $\mathcal{P}$ scales with replication factor $d$ in large-scale systems.
The first factor $w_{n, d}$ goes down with $d$, which increases both the lower and upper bound.
The impact of $d$ is more subtle in the second factor $Q_d(d \cdot x)$ where $x$ is $m$ and $2m$ for the lower and upper bound.
We next find a new lower and upper bound on $\mathcal{P}$, which are looser than those in \eqref{eq:P_asymptotic_bounds_for_cyclic_simpler} but more insightful on the impact of $d$.

\begin{corollary}
  Suppose the object demands $\rho_i$ are sub-gaussian as given in Theorem~\ref{thm:P_bounds_for_clustering}.
  Then, in $d$-choice allocation with cyclic design, there exists $\alpha, \beta, \gamma > 0$ such that in the limit $n \to \infty$
  \begin{equation}
  \begin{split}
    \mathcal{P} &\geq \exp\left(-w_{n, d} \cdot \gamma \cdot \exp\left(-d \cdot \beta(m - \mu)^2\right)\right) \\
    \mathcal{P} &\leq \exp\left(-w_{n, d} \cdot \exp\left(-d \cdot \alpha(2m - \mu)^2\right)\right),
  \end{split}
  \label{eq:P_asymptotic_bounds_for_cyclic_insightful}
  \end{equation}
  where $w_{n, d} = n - d + 1$ and $m \geq \mu$.
\label{cor:P_asymptotic_bounds_for_cyclic_insightful}
\end{corollary}
\begin{proof}
  See Appendix~\ref{subsec:proof_cor_P_asymptotic_bounds_for_cyclic_insightful}.
\end{proof}

\begin{remark}
  Similar to the bounds in Theorem~\ref{thm:P_bounds_for_clustering}, we can extend the bounds in \eqref{eq:P_asymptotic_bounds_for_cyclic_insightful} for demand distributions with heavier tail than sub-gaussian, such as sub-exponential distributions.

  The bounds in \eqref{eq:P_asymptotic_bounds_for_cyclic_insightful} say, for a fixed maximal load $m$, $\mathcal{P}$ for large-scale systems with cyclic design scales in $d$ as
  \[
    \mathcal{P} \;\myapprox\; \exp\left(-(n - d + 1) \cdot \exp\left(-d \cdot K\right)\right)
  \]
  for some constant $K > 0$.
  \closeremark
\end{remark}


Similar to \eqref{eq:P_scaling_with_d_n_m} given for clustering design, \eqref{eq:P_asymptotic_bounds_for_cyclic_insightful} allows us to analyze $\mathcal{P}$ for cyclic design as $d$ grows with $n$.

\begin{corollary}
    In $d$-choice allocation with cyclic design
    \begin{equation}
    \mathcal{P} \to
    \begin{cases} 
      1 & d = \Omega(\log(n)) \\
      0 & \text{otherwise}
    \end{cases}
    \quad\text{as } n \to \infty.
    \label{eq:P_limit_for_cyclic}
    \end{equation}
\label{cor:P_for_cyclic_as_d_to_infty}
\end{corollary}
\begin{proof}
  See Appendix~\ref{subsec:proof_cor_P_for_cyclic_as_d_to_infty}.
\end{proof}

\begin{remark}
  The limit in \eqref{eq:P_limit_for_cyclic} is the counterpart of the limit in Corollary~\ref{cor:P_for_clustering_as_d_to_infty}.
  Thus, observations noted for clustering design in the remark after Corollary~\ref{cor:P_for_clustering_as_d_to_infty} also hold for cyclic design.
  \closeremark
\end{remark}

\subsection{Block Design}
\label{subsec:block_design}
We refer the reader to Sec.~\ref{subsec:balanced_d_choice} for the description of the construction of $d$-choice allocation with block design.
Similar to the bounds given in Lemma~\ref{lm:P_bounds_for_rgap} for $r$-gap design, we can find the following bounds on $\mathcal{P}$ for systems with block design.
\begin{lemma}
  In $d$-choice allocation with block design
  \begin{equation}
     \Prob{S_d^{(c)} \leq m d/2} \leq \mathcal{P} \leq \Prob{S_d^{(c)} \leq m (d^2 - d)}.
  \label{eq:P_bounds_for_block_design}
  \end{equation}
\label{lm:P_bounds_for_block_design}
\end{lemma}
\begin{proof}
  See Appendix~\ref{subsec:proof_lm_P_bounds_for_block_design}.
\end{proof}

\begin{remark}
  Notice that the upper bound in \eqref{eq:P_bounds_for_block_design} scales with $d$ in the same manner as the bound in Corollary~\ref{cor:P_upper_bound_for_any_design_insightful} given for any storage design.
  Recall that the asymptotically tight bounds \eqref{eq:P_bounds_for_rgap} given for $r$-gap design scale more slowly with $d$ as 
  \[
    \mathcal{P} \leq \Prob{S_d^{(c)} \leq m \cdot (d + 2r)}.
  \]
  In this sense, block design can possibly achieve the best possible scaling of $\mathcal{P}$ in $d$. However, we cannot conclude this with our results. This is because the lower bound in \eqref{eq:P_bounds_for_block_design} given for block design scales in $d$ the same way as those given for $r$-gap design.
  Determining whether block design provides a better scaling of $\mathcal{P}$ in $d$ requires finding a tighter lower bound for block design, which is an open problem.
  \closeremark
\end{remark}

Just like the bounds given in Lemma~\ref{lm:P_bounds_for_cyclic} for cyclic design, the bounds given in Lemma~\ref{lm:P_bounds_for_block_design} allow us to use the results available on scan statistics to analyze and calculate $\mathcal{P}$ for block design.
Thus, similar to Theorem~\ref{thm:P_asymptotic_bounds_for_cyclic} for cyclic design, we next find asymptotic bounds on $\mathcal{P}$ for block design as follows.
\begin{theorem}
    In $d$-choice allocation with block design, in the limit $n \to \infty$ we have
    \begin{equation}
    \begin{split}
        \mathcal{P} &\geq \exp\left(-w_{n, d} \; Q_d(m d/2)\right), \\
        \mathcal{P} &\leq \exp\left(-w_{n, d} \; Q_d(m (d^2 - d))\right).
    \end{split}
    \label{eq:P_asymptotic_bounds_for_block}
    \end{equation}
    where $w_{n, d} = n - d + 1$ and $Q_d(x) = 1 - F_d(x)$.
\label{thm:P_asymptotic_bounds_for_block}
\end{theorem}
\begin{proof}
    See Appendix~\ref{subsec:proof_thm_P_asymptotic_bounds_for_block}.
\end{proof}

Bounds in \eqref{eq:P_asymptotic_bounds_for_block} show how $\mathcal{P}$ scales with $d$ in large-scale systems with block design.
Similar to what we did for cyclic design in Corollary~\ref{cor:P_asymptotic_bounds_for_cyclic_insightful}, we next find new asymptotic bounds on $\mathcal{P}$ for block design, which are looser than those in \eqref{eq:P_asymptotic_bounds_for_block} but more insightful on the impact of $d$.

\begin{corollary}
  Suppose the object demands $\rho_i$ are sub-gaussian as in Theorem~\ref{thm:P_bounds_for_clustering}.
  Then, in $d$-choice allocation with block design, there exists constants $\alpha, \beta, \gamma > 0$ such that in the limit $n \to \infty$
  \begin{equation}
  \begin{split}
    \mathcal{P} &\geq \exp\left(-w_{n, d} \cdot \gamma \cdot \exp\left(-d \cdot \beta(m/2 - \mu)^2\right)\right) \\
    \mathcal{P} &\leq \exp\left(-w_{n, d} \cdot \exp\left(-d \cdot \alpha(m(d - 1) - \mu)^2\right)\right),
  \end{split}
  \label{eq:P_asymptotic_bounds_for_block_insightful}
  \end{equation}
  where $w_{n, d} = n - d + 1$ and $m \geq \mu$.
\label{cor:P_asymptotic_bounds_for_block_insightful}
\end{corollary}
\begin{proof}
  See Appendix~\ref{subsec:proof_cor_P_asymptotic_bounds_for_block_insightful}.
\end{proof}

\begin{remark}
  Similar to the bounds in Theorem~\ref{thm:P_bounds_for_clustering} and Corollary~\ref{cor:P_asymptotic_bounds_for_cyclic_insightful}, we can extend \eqref{eq:P_asymptotic_bounds_for_block_insightful} for demand distributions with heavier tail than sub-gaussian, such as sub-exponential distributions.

  The bounds in \eqref{eq:P_asymptotic_bounds_for_block_insightful} show that, in large-scale systems with block design, $\mathcal{P}$ scales in $d$ \emph{at least as fast as} $r$-gap design, e.g., clustering and cyclic design. That is, in the limit $n \to \infty$
  \[
    \mathcal{P} \geq \exp\left(-(n - d + 1) \cdot \exp\left(-d \cdot K\right)\right)
  \]
  for some constant $K > 0$.
  The upper bound in \eqref{eq:P_asymptotic_bounds_for_block_insightful}, however, is not asymptotically tight. It says that in the limit $n \to \infty$
  \[
    \mathcal{P} \leq \exp\left(-(n - d + 1) \cdot \exp\left(-d^3 \cdot L\right)\right)
  \]
  for block design, where $L$ is a positive constant.
  \closeremark
\end{remark}



\subsection{Random Design}
\label{subsec:random_design}
As described in Sec.~\ref{subsec:balanced_d_choice}, with random design, each object is stored on a set of $d$ nodes chosen randomly. Recall that in the clustering and cyclic design, or $r$-gap design in general, service choice span for a given subset of objects is determined by the distance between the objects with respect to their indices. Instead of this deterministic structure, in random design, service choice spans are determined by the random node selection process that is independently performed for each object.


As the service choices $C_i$ are selected randomly and independently for each object, the span of $u$ objects is determined by the union of $u$ randomly chosen sets $C_i$ of fixed size $d$. This mathematical object is known as \emph{the occupancy metric for random allocation with complexes}.
\begin{definition}
  Let there be $n$ cells into which $u$ sets of $d$ particles (complexes) are thrown independently. Particles within each set are allocated to different cells. The number of cells containing at least one particle is defined as the \underline{occupancy} for random allocation with complexes and denoted as $N_{n, d, u}$.
\label{def:occupany_of_random_allocation_w_complexes}
\end{definition}

Using arguments similar to those we used to derive the bounds in Theorem~\ref{thm:P_upper_bound_for_any_storage_allocation}, Lemma~\ref{lm:P_bounds_for_cyclic} and Lemma~\ref{lm:P_bounds_for_block_design}, we can find an upper bound on $\mathcal{P}$ for random design as follows.
\begin{lemma}
  In $d$-choice allocation with random design
  \begin{equation}
    \mathcal{P} \leq \prod_{i = 1}^{v} \E_{N_{n, d, u_i}}\left[F_{u_i}(m \cdot N_{n, d, u_i})\right],
  \label{eq:P_upper_bound_for_random}
  \end{equation}
  where $v$ and $u_i$ are positive integers, and $u_1 + \dots + u_v = n$.
  $N_{n, d, u_i}$ is a random variable defined in Def.~\ref{def:occupany_of_random_allocation_w_complexes}.
\label{lm:P_upper_bound_for_random}
\end{lemma}
\begin{proof}
    See Appendix~\ref{subsec:proof_lm_P_upper_bound_for_random}.
\end{proof}

The upper bound in \eqref{eq:P_upper_bound_for_random} allows us to use the results on occupancy metric $N_{n, d, u}$ to analyze $\mathcal{P}$ for random design.
Occupancy for allocation with complexes is well-studied and much is known about its distribution. We here discuss only some of the results that we use, and refer the reader to \cite{RandomAllocations:KolchinSC78, RandomGraphs:Kolchin99, EstimatingSizeOfUnionOfRandomSubsets:BarotP01} and references therein for a more complete exposition.

We next present a version of \eqref{eq:P_upper_bound_for_random} in a form that is less general but more analytically insightful.
\begin{corollary}
  In $d$-choice allocation with random design
  \begin{equation}
    \mathcal{P} \leq \E_{N_{n, d, u}}\left[F_{u}(m \cdot N_{n, d, u})\right]^{n / u},
  \label{eq:P_upper_bound_for_random_w_even_partitioning}
  \end{equation}
  where $u$ is a positive integer.
\label{cor:P_upper_bound_for_random_w_even_partitioning}
\end{corollary}
\begin{proof}
    See Appendix~\ref{subsec:proof_cor_P_upper_bound_for_random_w_even_partitioning}.
\end{proof}

In Chapter~7 of \cite{RandomAllocations:KolchinSC78}, an exact expression is given for the distribution of occupancy metric $N_{n, d, u}$. We can use this expression to calculate the upper bound in \eqref{eq:P_upper_bound_for_random_w_even_partitioning}. However, this would not yield a tractable expression, hence would not give insight on how the upper bound in \eqref{eq:P_upper_bound_for_random_w_even_partitioning} scales with $d$.
One approach to obtain a tractable expression would be approximating the random variable $N_{n, d, u}$ with its average value. That gives the following approximate upper bound
\begin{equation*}
    \mathcal{P} \lesssim F_u\left(m \cdot \Exp{N_{n, d, u}}\right)^{n / u},
\end{equation*}
for any positive integer $u$.
Setting $u = d$ above gives us
\begin{equation}
    \mathcal{P} \lesssim F_d\left(m \cdot \Exp{N_{n, d, d}}\right)^{n / d}.
\label{eq:P_approx_upper_bound_for_block_for_u_eq_d}
\end{equation}
Notice that this form of the upper bound is the same as the exact expression of $\mathcal{P}$ given for clustering design in \eqref{eq:P_clustering}, except that the argument $m \cdot d$ of $F_d$ is replaced here by $m \cdot \Exp{N_{n, d, d}}$. Using the expression given in Chapter 7 of \cite{RandomAllocations:KolchinSC78}, we get
\[
    \Exp{N_{n, d, d}} = n\left(1 - (1 - d / n)^d\right).
\]
Using this, we can show that
$\lim_{n \to \infty} \Exp{N_{n, d, d}} = d^2$,
which implies in the limit $n \to \infty$
\begin{equation}
    \mathcal{P} \lesssim F_d\left(m \cdot d^2\right)^{n / d}.
\label{eq:P_limiting_approx_upper_bound_for_block_for_u_eq_d}
\end{equation}
This, together with the arguments we used in Theorem~\ref{thm:P_bounds_for_clustering} gives us the following approximate upper bound in the limit $n \to \infty$
\begin{equation}
    \mathcal{P} \lesssim \left(1 - \exp(-d \cdot \alpha (m \cdot d - \mu)^2)\right)^{n / d}.
\label{eq:P_limiting_approx_upper_bound_for_block_for_u_eq_d_refined}
\end{equation}

Comparing \eqref{eq:P_limiting_approx_upper_bound_for_block_for_u_eq_d} with \eqref{eq:P_clustering} shows that \eqref{eq:P_limiting_approx_upper_bound_for_block_for_u_eq_d} is an upper bound for clustering design also.
Thus by using random design, we can possibly achieve a better scaling of $\mathcal{P}$ in replication factor $d$ compared to clustering or cyclic design. 
Recall from Sec.~\ref{subsec:clustering_design} and \ref{subsec:cyclic_design} that $\mathcal{P} \to 0$ as $n \to \infty$ when $d = o(\log(n))$.
By \eqref{eq:P_limiting_approx_upper_bound_for_block_for_u_eq_d_refined}, we may have $\mathcal{P} \to 1$ when $d = \Omega(\log(n)^{1/3})$.
However, we cannot conclude this as we do not have a tight lower bound on $\mathcal{P}$ for random design. We leave this as an open problem. In the remainder, we discuss an approach for adding structure in random design so that we can derive a lower bound for $\mathcal{P}$.

\vspace{1ex}
\noindent
\textbf{Lower bound on $\mathcal{P}$}:\space
For other storage designs discussed so far, we found a lower bound on $\mathcal{P}$ by considering a set of ``worst-case'' demand vectors for service and showing that the system can meet the maximal load requirement if the demand vectors obey a particular sufficiency rule. Scan statistic has proved to be an effective tool to capture this sufficiency rule and obtain an asymptotically tight lower bound for $\mathcal{P}$ (recall Lemma~\ref{lm:P_bounds_for_rgap} and Corollary~\ref{cor:P_for_cyclic_as_d_to_infty}). The \emph{invariant structure implemented in service choice overlaps} made this possible for scan statistics. 

Random design, however, does not impose any structure on service choice overlaps. 
Even the extreme event of all objects getting assigned to the same set of $d$ nodes can happen with a non-zero probability. This makes it impossible to find sufficient conditions to meet the maximal load requirement that would give a lower bound for $\mathcal{P}$.
We next introduce a \emph{constrained random design} that limits the random node selection process. The goal with this is to instill enough structure in service choice overlaps to find a lower bound for $\mathcal{P}$.
For a given object $i$, let us define its \emph{$d$-hop siblings} as the following set
\[
    \Psi_i = \set{o_j \;\bigl|\; |j - i| = 0 \!\!\!\mod d, \;j \neq i, \;0 \leq j \leq n}.
\]
For a given object $i$, let us also define the number of $d$-hop siblings with which it has non-zero service choice overlaps as
\[
    v_i = \card{\set{o_j \;\bigl|\; o_j \in \Psi_i, \;\card{C_i \cap C_j} > 0}}.
\]
Let us also define a limit on the number of overlapping siblings as $v_i \leq v_{\max}$ for all $i$.
Then, we select $d$ nodes for each object by randomly drawing from all \emph{suitable} nodes without replacement. Unlike the random design (as described in Sec.~\ref{subsec:balanced_d_choice}), not all nodes are suitable for each object. A node is deemed suitable for an object only if storing the object on this node would maintain the limit $v_i \leq v_{\max}$.
Notice that this constraint limits the overlapping objects similar to $r$-gap design, but the constraint is more relaxed.

Using the structure implemented by constrained random design, we can find a lower bound for $\mathcal{P}$ as follows.
\begin{lemma}
  In $d$-choice allocation with constrained random design with maximum number of overlapping siblings $v_{\max}$
  \begin{equation}
    \mathcal{P} \geq \Prob{S_d^{(c)} \leq m d / v_{\max}}.
  \label{eq:P_lower_bound_for_constrained_random}
  \end{equation}
\label{lm:P_lower_bound_for_constrained_random}
\end{lemma}
\begin{proof}
  See Appendix~\ref{subsec:proof_lm_P_lower_bound_for_constrained_random}.
\end{proof}
\begin{remark}
  Notice that the lower bound \eqref{eq:P_lower_bound_for_constrained_random} is the same as the lower bound given for cyclic design in \eqref{eq:P_bounds_for_cyclic}, except the dividing factor $v_{\max}$. Thus, we can derive asymptotic lower bounds for constrained random design by employing the same arguments we used for cyclic design between Theorem~\ref{thm:P_asymptotic_bounds_for_cyclic} and Corollary~\ref{cor:P_for_cyclic_as_d_to_infty}.
  Given this, if $v_{\max}$ stays finite as the system scale $n \to \infty$, we can conclude for contrained random design that $\mathcal{P} \to 1$ as $n \to \infty$ if and only if $d = \Omega(\log(n))$.
  \closeremark
\end{remark}




\section{Conclusions}
\label{sec:conclusions}
Distributed systems replicate data objects to balance the offered load across the storage nodes. Offered load should be as evenly distributed over the nodes as possible in order to provide fast and predictable data access performance. Furthermore, load balancing should be robust in the presence of skews and changes in object popularities. 

Load balancing performance is mainly determined by two design decisions: (1) the replication factor, (2) the assignment of objects to storage nodes. In this paper, we analyzed the performance implications of these two design choices by considering four storage schemes used in practical systems: clustering, cyclic, block and random design. In our analysis, we consider the goal of load balancing as maintaining the load on any node below a given threshold.

For the first design decision, we derived necessary and sufficient conditions to achieve the desired load balance in a system of $n$ nodes as $n \to \infty$. First, we found that the replication factor $d = \Omega(\log(n)^{1/3})$ is necessary for any storage scheme. When the overlaps between object copies are few but occasionally large (i.e., clustering and cyclic design), we found that $d = \Omega(\log(n))$ is necessary and sufficient. When the overlaps are many but consistently small (i.e., block and random design), we found that $d = \Omega(\log(n))$ is sufficient but not necessary. This overall implies that $d = \Omega(\log(n)^{1/3})$ might be sufficient. We do not have a proof for this and left it as an open problem.

For the second design decision, we showed that the choice of storage scheme depends on the level of skew in object demands.
We found that, in majority of the cases, many but consistently small overlaps between object copies is better than few but occasionally large overlaps, i.e., $\text{block} > \text{random} > \text{clustering} > \text{cyclic}$ where $>$ denotes ``better''. However, the impact of overlaps becomes the opposite when the level of skew goes beyond a level.

\section{Acknowledgements}
This research is supported by the National Science Foundation
under Grants No. CIF-2122400.


\bibliographystyle{unsrt}
\bibliography{references}

\begin{thebibliography}{10}

\bibitem{ChallengesInBuildingLargeScaleInformationRetrievalSystems:Dean09}
Jeffrey Dean.
\newblock Challenges in building large-scale information retrieval systems:
  invited talk.
\newblock In Ricardo Baeza{-}Yates, Paolo Boldi, Berthier~A. Ribeiro{-}Neto,
  and Berkant~Barla Cambazoglu, editors, {\em Proceedings of the Second
  International Conference on Web Search and Web Data Mining, {WSDM} 2009,
  Barcelona, Spain, February 9-11, 2009}, page~1. {ACM}, 2009.

\bibitem{Dremel:MelnikGL10}
Sergey Melnik, Andrey Gubarev, Jing~Jing Long, Geoffrey Romer, Shiva
  Shivakumar, Matt Tolton, and Theo Vassilakis.
\newblock Dremel: Interactive analysis of web-scale datasets.
\newblock {\em Proc. {VLDB} Endow.}, 3(1):330--339, 2010.

\bibitem{TailAtScale:DeanB13}
Jeffrey Dean and Luiz~Andr{\'{e}} Barroso.
\newblock The tail at scale.
\newblock {\em Commun. {ACM}}, 56(2):74--80, 2013.

\bibitem{StragglerRootCauseAnalysisInDatacenters:OuyangGY16}
Xue Ouyang, Peter Garraghan, Renyu Yang, Paul Townend, and Jie Xu.
\newblock Reducing late-timing failure at scale: Straggler root-cause analysis
  in cloud datacenters.
\newblock In {\em Fast Abstracts in the 46th Annual IEEE/IFIP International
  Conference on Dependable Systems and Networks}, 2016.

\bibitem{HDFS:ShvachkoKR10}
Konstantin Shvachko, Hairong Kuang, Sanjay Radia, and Robert Chansler.
\newblock The hadoop distributed file system.
\newblock In Mohammed~G. Khatib, Xubin He, and Michael Factor, editors, {\em
  {IEEE} 26th Symposium on Mass Storage Systems and Technologies, {MSST} 2012,
  Lake Tahoe, Nevada, USA, May 3-7, 2010}, pages 1--10. {IEEE} Computer
  Society, 2010.

\bibitem{Cassandra:LakshmanM10}
Avinash Lakshman and Prashant Malik.
\newblock Cassandra: a decentralized structured storage system.
\newblock {\em {ACM} {SIGOPS} Oper. Syst. Rev.}, 44(2):35--40, 2010.

\bibitem{Redis}
Salvatore Sanfilippo.
\newblock {\em Redis: An open source (BSD licensed), in-memory data structure
  store.}, 2023.

\bibitem{Scarlett:AnanthanarayananAK11}
Ganesh Ananthanarayanan, Sameer Agarwal, Srikanth Kandula, Albert~G. Greenberg,
  Ion Stoica, Duke Harlan, and Ed~Harris.
\newblock Scarlett: Coping with skewed content popularity in mapreduce
  clusters.
\newblock In Christoph~M. Kirsch and Gernot Heiser, editors, {\em European
  Conference on Computer Systems, Proceedings of the Sixth European Conference
  on Computer systems, EuroSys 2011, Salzburg, Austria, April 10-13, 2011},
  pages 287--300. {ACM}, 2011.

\bibitem{BallsIntoBins:RaabS98}
Martin Raab and Angelika Steger.
\newblock "balls into bins" - {A} simple and tight analysis.
\newblock In Michael Luby, Jos{\'{e}} D.~P. Rolim, and Maria~J. Serna, editors,
  {\em Randomization and Approximation Techniques in Computer Science, Second
  International Workshop, RANDOM'98, Barcelona, Spain, October 8-10, 1998,
  Proceedings}, volume 1518 of {\em Lecture Notes in Computer Science}, pages
  159--170. Springer, 1998.

\bibitem{BalancedAllocations:AzarBK99}
Yossi Azar, Andrei~Z. Broder, Anna~R. Karlin, and Eli Upfal.
\newblock Balanced allocations.
\newblock {\em {SIAM} J. Comput.}, 29(1):180--200, 1999.

\bibitem{BalancedAllocations_HeavilyLoadedCase:Berenbrink20}
Petra Berenbrink, Artur Czumaj, Angelika Steger, and Berthold V{\"{o}}cking.
\newblock Balanced allocations: the heavily loaded case.
\newblock In F.~Frances Yao and Eugene~M. Luks, editors, {\em Proceedings of
  the Thirty-Second Annual {ACM} Symposium on Theory of Computing, May 21-23,
  2000, Portland, OR, {USA}}, pages 745--754. {ACM}, 2000.

\bibitem{BalancedAllocationsOnGraphs:Kenthapadi06}
Krishnaram Kenthapadi and Rina Panigrahy.
\newblock Balanced allocation on graphs.
\newblock In {\em Proceedings of the Seventeenth Annual {ACM-SIAM} Symposium on
  Discrete Algorithms, {SODA} 2006, Miami, Florida, USA, January 22-26, 2006},
  pages 434--443. {ACM} Press, 2006.

\bibitem{BalancedAllocationsOnHypergraphs:Godfrey08}
Brighten Godfrey.
\newblock Balls and bins with structure: balanced allocations on hypergraphs.
\newblock In Shang{-}Hua Teng, editor, {\em Proceedings of the Nineteenth
  Annual {ACM-SIAM} Symposium on Discrete Algorithms, {SODA} 2008, San
  Francisco, California, USA, January 20-22, 2008}, pages 511--517. {SIAM},
  2008.

\bibitem{BatchCodesAndTheirApps:IshaiKO04}
Yuval Ishai, Eyal Kushilevitz, Rafail Ostrovsky, and Amit Sahai.
\newblock Batch codes and their applications.
\newblock In L{\'{a}}szl{\'{o}} Babai, editor, {\em Proceedings of the 36th
  Annual {ACM} Symposium on Theory of Computing, Chicago, IL, USA, June 13-16,
  2004}, pages 262--271. {ACM}, 2004.

\bibitem{CombinatorialBatchCodes:StinsonWP09}
Maura~B. Paterson, Douglas~R. Stinson, and Ruizhong Wei.
\newblock Combinatorial batch codes.
\newblock {\em Adv. Math. Commun.}, 3(1):13--27, 2009.

\bibitem{AllertonServiceCapacity:AktasJS17}
Mehmet Akta{\c{s}}, Sarah~E. Anderson, Ann Johnston, Gauri Joshi, Swanand
  Kadhe, Gretchen~L. Matthews, Carolyn Mayer, and Emina Soljanin.
\newblock On the service capacity region of accessing erasure coded content.
\newblock In {\em 55th Annual Allerton Conference on Communication, Control,
  and Computing, Allerton 2017, Monticello, IL, USA, October 3-6, 2017}, pages
  17--24. {IEEE}, 2017.

\bibitem{ServiceRateRegion:AktasJK21}
Mehmet Akta{\c{s}}, Gauri Joshi, Swanand Kadhe, Fatemeh Kazemi, and Emina
  Soljanin.
\newblock Service rate region: A new aspect of coded distributed system design.
\newblock {\em IEEE Transactions on Information Theory}, 67(12):7940--7963,
  2021.

\bibitem{ITWServiceCapacity:AndersonJJ18}
Sarah~E. Anderson, Ann Johnston, Gauri Joshi, Gretchen~L. Matthews, Carolyn
  Mayer, and Emina Soljanin.
\newblock Service rate region of content access from erasure coded storage.
\newblock In {\em {IEEE} Information Theory Workshop, {ITW} 2018, Guangzhou,
  China, November 25-29, 2018}, pages 1--5. {IEEE}, 2018.

\bibitem{LoadBalancing:AktasFS21}
Mehmet~Fatih Akta{\c{s}}, Amir~Behruzi Far, Emina Soljanin, and Philip Whiting.
\newblock Evaluating load balancing performance in distributed storage with
  redundancy.
\newblock {\em IEEE Transactions on Information Theory}, 67(6):3623--3644,
  2021.

\bibitem{ScanStats:GlazNW01}
Joseph Glaz, Joseph~I Naus, Sylvan Wallenstein, Sylvan Wallenstein, and
  Joseph~I Naus.
\newblock {\em Scan statistics}.
\newblock Springer, 2001.

\bibitem{RandomAllocations:KolchinSC78}
Valentin~Fedorovich Kolchin, Boris~A Sevast'janov, and Vladimir~P Christ'Yakov.
\newblock {\em Random allocations}.
\newblock Vh Winston, 1978.

\bibitem{ParetoStochasticDominance:ChenEW22}
Yuyu Chen, Paul Embrechts, and Ruodu Wang.
\newblock An unexpected stochastic dominance: Pareto distributions,
  catastrophes, and risk exchange.
\newblock {\em arXiv preprint arXiv:2208.08471}, 2022.

\bibitem{CombinatorialDesigns:Stinson07}
Douglas~R. Stinson.
\newblock {\em Combinatorial designs - constructions and analysis}.
\newblock Springer, 2004.

\bibitem{Scylladb:Suneja19}
Nishant Suneja.
\newblock Scylladb optimizes database architecture to maximize hardware
  performance.
\newblock {\em IEEE Software}, 36(04):96--100, 2019.

\bibitem{ApproxForScanStatDist:Naus82}
Joseph~I Naus.
\newblock Approximations for distributions of scan statistics.
\newblock {\em Journal of the American Statistical Association},
  77(377):177--183, 1982.

\bibitem{LimitLawsOfMovingSum:DeheuvelsD87}
Paul Deheuvels and Luc Devroye.
\newblock Limit laws of erdos-renyi-shepp type.
\newblock {\em The Annals of Probability}, pages 1363--1386, 1987.

\bibitem{PoissonApproxForScanStat:DemboK92}
Amir Dembo and Samuel Karlin.
\newblock Poisson approximations for r-scan processes.
\newblock {\em The Annals of Applied Probability}, pages 329--357, 1992.

\bibitem{ExtremeValueMethodsWithAppsToFinance:Novak11}
Serguei~Y Novak.
\newblock {\em Extreme value methods with applications to finance}.
\newblock CRC press, 2011.

\bibitem{RandomGraphs:Kolchin99}
Valentin~Fedorovich Kolchin.
\newblock {\em Random graphs}.
\newblock Number~53 in Encyclopedia of mathematics and its applications.
  Cambridge University Press, 1999.

\bibitem{EstimatingSizeOfUnionOfRandomSubsets:BarotP01}
Michael Barot and Jos{\'e}~Antonio de~la Pena.
\newblock Estimating the size of a union of random subsets of fixed
  cardinality.
\newblock {\em Elemente der Mathematik}, 56:163--169, 2001.

\bibitem{NonAsymptoticLowerTailBounds:ZhangZ20}
Anru~R Zhang and Yuchen Zhou.
\newblock On the non-asymptotic and sharp lower tail bounds of random
  variables.
\newblock {\em Stat}, 9(1):e314, 2020.

\end{thebibliography}

\newpage
\section{Appendix}
\label{sec:appendix}

\subsection{Proof of Lemma~\ref{lm:cap_region_for_system_w_replicated_storage}}
\label{subsec:proof_lm_cap_region_for_replicated_storage}
Let us start by restating the claim in a slightly different but equivalent way: 
for a demand vector $(\rho_1, \dots, \rho_n)$ to lie in the capacity region, it is \emph{necessary} and \emph{sufficient} to have $\sum_{i \in I}{\rho_i} \leq \mathrm{span}(I)$ for all $I \subset \{1, \ldots, n\}$. 
We here prove this claim in two steps as given below.

\vspace{1ex}
\noindent
\textbf{Necessary condition:} Span of a set of objects gives the total capacity available to serve those objects jointly.
It is surely impossible to stabilize the system when the cumulative demand for a set of objects exceeds the cumulative capacity available to serve them. Therefore, for the system to serve the demand with stability, the span of a set of objects must be at least as much as the cumulative demand for those objects.

If a system can serve a given demand vector $\bm{\rho}$, then the system would be able to serve any demand vector that is constructed by setting a subset of the demands within $\bm{\rho}$ to zero. That is, if the system can serve $\bm{\rho} = (\rho_1, \dots, \rho_n)$, then it can serve all demand vectors $\bm{\rho}^{I}$ for $I \subset \{1, \ldots, n\}$ such that $\rho_i^{I} = \rho_i$ for $i \in I$ and $\rho_i^{I} = 0$ for $i \not\in I$.
For instance, if a system can serve $(1, 2, 3)$, then it can also serve all of its smaller subsets: $(1, 0, 0)$, $(0, 2, 0)$, $(0, 0, 3)$, $(1, 2, 0)$, $(1, 0, 3)$, $(0, 2, 3)$. This is why it is necessary for the system to have sufficient cumulative capacity available to serve all subsets of objects $\{o_i \mid i \in I\}$ for $I \subset \{1, \ldots, k\}$ with the cumulative demand $\sum_{i \in I}{\rho_i}$. Thus the set of conditions given in \eqref{eq:cap_region_w_service_choice_spans} is necessary for system stability.

\vspace{1ex}
\noindent
\textbf{Sufficient condition:}
Let us first consider the case with all object demands being zero except for \emph{one} object, say $o_1$ without loss of generality. It is then easy to see that the system can serve the demand if $\rho_1 \leq \mathrm{span}(\{1\}) = \lvert C_1 \rvert$. 

Let us next consider the case with all object demands being zero except for two objects, say $o_1$ and $o_2$ without loss of generality. System can serve the demand if two conditions are met. First, the object demands can be served individually, that is,
\begin{equation}
    \rho_1 \leq \lvert C_1 \rvert, \quad \rho_2 \leq \lvert C_2 \rvert.
\label{eq:serving_two_objs_cond_1}
\end{equation}
Second, the object demands can be served jointly. It is easy to find the two extreme points at which the object demands can be served jointly: (1) $\rho_1 = \lvert C_1 \rvert, ~\rho_2 = \lvert C_2 - C_1 \rvert$, (2) $\rho_1 = \lvert C_1 - C_2 \rvert, ~\rho_2 = \lvert C_2 \rvert$. Given that the capacity region is a convex polytope, the system can serve any convex combination of these two extreme points, which implies that the system is stable if
\begin{equation}
    \rho_1 + \rho_2 \leq \lvert C_1 \cup C_2 \rvert.
\label{eq:serving_two_objs_cond_2}
\end{equation}
Putting \ref{eq:serving_two_objs_cond_1} and \ref{eq:serving_two_objs_cond_2} together, we conclude that the system can serve the demand if
\begin{equation}
    \rho_1 \leq \lvert C_1 \rvert, \quad \rho_2 \leq \lvert C_2 \rvert, \quad \rho_1 + \rho_2 \leq \lvert C_1 \cup C_2 \rvert.
\label{eq:serving_two_objs_cond}
\end{equation}
Note that the capacity at the nodes in $C_1 \cup C_2$ will be fully used to serve the demand if $\rho_1 + \rho_2 = \lvert C_1 \cup C_2 \rvert$. This observation will be useful for the next step.

Let us now consider the case with object demands being zero except for three objects, say $o_1$, $o_2$ and $o_3$ without loss of generality. We already know from above that the condition given in \ref{eq:serving_two_objs_cond} needs to be met for serving each pair of objects $\{o_1, o_2\}$, $\{o_1, o_3\}$ or $\{o_2, o_3\}$ jointly. We now need to find the sufficient condition for serving all the three objects $\{o_1, o_2, o_3\}$ jointly. It is easy to find the three extreme points at which the object demands can be served jointly:
\begin{itemize}[label={}]
    \item $\rho_1 = \lvert C_1 - (C_2 \cup C_3) \rvert, ~\rho_2 + \rho_3 = \lvert C_2 \cup C_3 \rvert$,
    \item $\rho_2 = \lvert C_2 - (C_1 \cup C_3) \rvert, ~\rho_1 + \rho_3 = \lvert C_1 \cup C_3 \rvert$,
    \item $\rho_3 = \lvert C_3 - (C_1 \cup C_2) \rvert, ~\rho_1 + \rho_2 = \lvert C_1 \cup C_2 \rvert$.
\end{itemize}
Given that the capacity region is a convex polytope, the system can serve any convex combination of these three extreme points, which implies that the system is stable if
\begin{equation}
    \rho_1 + \rho_2 + \rho_2 \leq \lvert C_1 \cup C_2 \cup C_2 \rvert.
\label{eq:serving_three_objs_cond_1}
\end{equation}
Putting all conditions together, we conclude that the system can serve the demand if
\begin{gather*}
    \rho_1 \leq \lvert C_1 \rvert, ~ \rho_2 \leq \lvert C_2 \rvert, ~ \rho_3 \leq \lvert C_3 \rvert, \\
    \rho_1 + \rho_2 \leq \lvert C_1 \cup C_2 \rvert, \\ \rho_1 + \rho_3 \leq \lvert C_1 \cup C_3 \rvert, \\ \rho_2 + \rho_3 \leq \lvert C_2 \cup C_3 \rvert, \\
    \rho_1 + \rho_2 + \rho_3 \leq \lvert C_1 \cup C_2 \cup C_3 \rvert.
\label{eq:serving_three_objs_cond}
\end{gather*}

Similar to above, we can continue incrementing the number of objects with non-zero demand until covering all $n$ objects. Using the same arguments outlined above, when all $n$ objects have non-zero demand, the system can serve the demand if $\sum_{i \in I}{\rho_i} \leq \mathrm{span}(I)$ for all $I \subset \{1, \ldots, n\}$. Thus the set of conditions given in \eqref{eq:cap_region_w_service_choice_spans} is sufficient for system stability.

\subsection{Proof of Corollary~\ref{cor:P_larger_demand_leq_P}}
\label{subsec:proof_cor_P_larger_demand_leq_P}
Recall from \eqref{eq:P_integration} that $\mathcal{P}$ can be written as the integration of the joint density of the object demands over the modified capacity region $\mathcal{C}_m$ as
\[
    \mathcal{P} = \int_{\bm{c} \in \mathcal{C}_m} \Prob{\bm{\rho} = \bm{c}} \diff \bm{c}.
\]
Recall also from \eqref{eq:P_integrand} that object demands $\rho_i$ are i.i.d., hence
\[
    \Prob{\bm{\rho} = \bm{c}} = \prod_{i = 1}^{k} f(c_i).
\]
The above expressions also hold for the demand vector $\bm{\rho}^\prime$ with demand distribution $F^\prime$.

Now recall from Sec.~\ref{subsec:cap_region} that $\mathcal{C}_m$ is a convex polytope, which is defined in the non-negative orthant and it contains the origin (zeros vector).
Let us denote the boundary of $\mathcal{C}_m$ as $\mathcal{B}_m$. Note that the maximal access load in the system is equal to $m$ under the demand vectors in $\mathcal{B}_m$. Each point $\bm{b}$ in $\mathcal{B}_m$ defines a line segment $l_{\bm{b}}$ to the origin as
\[
    l_{\bm{b}} = \{\alpha \cdot \bm{b}\; |\; 0 \leq \alpha \leq 1\}.
\]
We can then cover all vectors in $\mathcal{C}_m$ by line segments $l_{\bm{\rho}}$, which allows us to express the integral given above for $\mathcal{P}$ as
\begin{equation}
    \mathcal{P} = \int_{\bm{b} \in \mathcal{B}_m} \int_{0}^{1} \Prob{\bm{\rho} = \alpha \cdot \bm{b}} \diff \alpha \diff \bm{b}.
\label{eq:P_integral_w_line_segments}
\end{equation}
where the inner integral is over a given line segment $l_{\bm{\rho}}$.

The ordering $F^\prime(x) < F(x)$ implies that demand vectors $\bm{\rho}$ with smaller coordinate values $\rho_i$ become less likely when the demand distribution $F(x)$ is replaced with $F^\prime(x)$. This implies that the inner integral in \eqref{eq:P_integral_w_line_segments} becomes smaller when $\rho$ is replaced with $\rho^\prime$. This shows that $\mathcal{P} > \mathcal{P}^\prime$.

\subsection{Proof of Lemma~\ref{lm:balancing_allocation_reduces_cum_overlap}}
\label{subsec:proof_lm_balancing_allocation_reduces_cum_overlap}
Consider a fixed value of $t$ in $[1, u)$.
Let us denote the nodes with $u$ and $v$ objects as node-$u$ and node-$v$ respectively.
Let us now pick one of the object copies on node-$u$, which we refer to as the tagged object.
The tagged object participates in ${u - 1 \choose t - 1}$ many $t$-subsets of objects stored on node-$u$. Hence it contributes ${u - 1 \choose t - 1}$ overlaps to the cumulative overlap $\mathrm{CumOverlap}_t$.

Following the statement given in the Lemma, let us now move the tagged object from node-$u$ to node-$v$. The tagged object will no longer overlap with the objects on node-$u$ hence its contribution to the cumulative overlap due to its overlap with objects on node-$u$ is now zero. However, it now overlaps with the objects stored on node-$v$ and contributes ${v \choose t - 1}$ overlaps to the cumulative overlap. Given that $u > v + 1$, we have
\[
    {v \choose t - 1} > {u - 1 \choose t - 1}.
\]
This shows that moving an object from node-$u$ to node-$v$ reduces $\mathrm{CumOverlap}_t$ by ${v \choose t - 1} - {u - 1 \choose t - 1}$.

\subsection{Proof of Lemma~\ref{lm:balanced_d_choice_allocation_cum_overlap}}
\label{subsec:proof_lm_balanced_d_choice_allocation_cum_overlap}
In Lemma~\ref{lm:balancing_allocation_reduces_cum_overlap}, we have shown that moving an object from one node to another with fewer objects reduces the cumulative overlaps.
Using the same arguments used in Lemma~\ref{lm:balancing_allocation_reduces_cum_overlap}, we can show that the opposite is also true: moving an object from one node to another with more objects increases the cumulative overlaps. This implies that the balanced allocation minimizes the cumulative overlap.

We next find the value of cumulative overlap $\mathrm{CumOverlap}_t$ for a given $t$.
In a balanced $d$-choice allocation, each node stores $d$ different object copies. On every node, there are ${d \choose t}$ many $t$-subsets overlapping on the node. This together with the fact that there are $n$ nodes in the system gives us the expression for cumulative overlap as given in \eqref{eq:cum_overlap_in_balanced_d_choice}.

\subsection{Approximating block design with random construction}
\label{subsec:approx_block_design_w_random_construction}
As described in Section~\ref{subsec:balanced_d_choice}, block design exists only for system parameters meeting the following equality: $n = d^2 - d + 1$.
It is however possible to construct storage designs that approximate block design. We here give an example randomized construction that we used in this paper while evaluating $\mathcal{P}$ numerically for block design.

Our randomized storage construction is given as follows.
We start by putting $d$ copies for each object into a list. We then shuffle the list and put the object copies into a queue that we refer to as \emph{object queue}.
After this, all we do is essentially pulling the object copies one by one from the queue and placing each on a randomly selected node.
However due to random selection of nodes, for some of the objects we might end up choosing a node that 
(1) already stores a copy of the object or 
(2) already stores $d$ different object copies (recall that each node stores $d$ different objects in a $d$-choice allocation).
If we hit one of these two cases, we continue looking for a viable node for the object pulled from the queue until we find one, which we refer to as \emph{object at hand}.
We first try by incrementing the node index until we land on a viable node. 
If we still fail to find a viable node, we go back to the first node and start iterating over all the nodes until we land on a node (tagged node) that does not store a copy for the object. We then pick one of the objects stored on this tagged node at random, remove it from the node and put it back in the object queue. This makes the node viable for the object at hand and we place it there. Putting the removed object copy back into the object queue also makes sure that we will find a viable node for it when we pull it from the queue next time.

Recall that block design distributes the overlaps evenly between all pairs of service choices. The construction described above obviously does not guarantee this even when the system parameters allow for block design. It does not even attempt to distribute overlaps evenly across the service choices, but merely assigns the object copies to nodes using a ``guided'' random process. However, we find that it tends to achieve small overlaps between the service choices.
Table~\ref{table:overlap_distribution} shows the distribution of overlaps computed numerically. In a single \emph{run}, we use the construction given above and create storage design instance for a given $n$ and $d$, and compute the overlap distribution for the design. For each $n$ and $d$, we execute $100$ runs, and report the average and standard deviation values for each overlap size.

We draw two conclusions from Table~\ref{table:overlap_distribution}.
First, the randomized construction given above tends to minimize the overlaps between the service choices. Majority of the overlaps are of size $1$, and the fraction of larger overlaps goes down quickly (similar to Zipf's Law) as the overlap size increases.
As the replication factor $d$ increases, the overlaps tend to get larger.
This is because, when $n$ is fixed, larger $d$ increases the chance of overlaps. This is why block design requires quadratically larger number of nodes ($d^2 + d - 1$) as $d$ increases.
In practice, $d$ is desired to be and is typically much smaller than $n$, in which case the randomized construction tends to be a very good approximator of the block design.
As shown in Table~\ref{table:overlap_distribution}, when $d < 0.05 n$, the randomized construction yields an overlap of size $1$ between more than $\myapprox 94\%$ of the service choices.

\begin{table*}[h]
\centering
\begin{tabular}{cccccccccccc}
                                            &                         & \multicolumn{10}{c}{Overlap size}                                                                                                                                                      \\ \cline{3-12} 
                                            & \multicolumn{1}{c|}{}   & \multicolumn{2}{c|}{1}             & \multicolumn{2}{c|}{2}             & \multicolumn{2}{c|}{3}             & \multicolumn{2}{c|}{4}             & \multicolumn{2}{c|}{5}             \\ \hline
\multicolumn{1}{|c|}{n}                     & \multicolumn{1}{c|}{d}  & mean   & \multicolumn{1}{c|}{stdev} & mean   & \multicolumn{1}{c|}{stdev} & mean   & \multicolumn{1}{c|}{stdev} & mean   & \multicolumn{1}{c|}{stdev} & mean   & \multicolumn{1}{c|}{stdev} \\ \hline
\multicolumn{1}{|c|}{\multirow{9}{*}{100}}  & \multicolumn{1}{c|}{2}  & 0.997 & \multicolumn{1}{c|}{0.004} & 0.010 & \multicolumn{1}{c|}{0}     &       & \multicolumn{1}{c|}{}      &       & \multicolumn{1}{c|}{}      &       & \multicolumn{1}{c|}{}      \\
\multicolumn{1}{|c|}{}                      & \multicolumn{1}{c|}{3}  & 0.986 & \multicolumn{1}{c|}{0.007} & 0.014 & \multicolumn{1}{c|}{0.007} & 0.003 & \multicolumn{1}{c|}{0}     &       & \multicolumn{1}{c|}{}      &       & \multicolumn{1}{c|}{}      \\
\multicolumn{1}{|c|}{}                      & \multicolumn{1}{c|}{4}  & 0.966 & \multicolumn{1}{c|}{0.008} & 0.034 & \multicolumn{1}{c|}{0.008} & 0.002 & \multicolumn{1}{c|}{0}     &       & \multicolumn{1}{c|}{}      &       & \multicolumn{1}{c|}{}      \\
\multicolumn{1}{|c|}{}                      & \multicolumn{1}{c|}{5}  & 0.935 & \multicolumn{1}{c|}{0.006} & 0.063 & \multicolumn{1}{c|}{0.006} & 0.002 & \multicolumn{1}{c|}{0.001} & 0.001 & \multicolumn{1}{c|}{0}     &       & \multicolumn{1}{c|}{}      \\
\multicolumn{1}{|c|}{}                      & \multicolumn{1}{c|}{6}  & 0.890 & \multicolumn{1}{c|}{0.006} & 0.104 & \multicolumn{1}{c|}{0.006} & 0.005 & \multicolumn{1}{c|}{0.001} & 0.001 & \multicolumn{1}{c|}{0}     &       & \multicolumn{1}{c|}{}      \\
\multicolumn{1}{|c|}{}                      & \multicolumn{1}{c|}{7}  & 0.842 & \multicolumn{1}{c|}{0.008} & 0.145 & \multicolumn{1}{c|}{0.008} & 0.012 & \multicolumn{1}{c|}{0.002} & 0.001 & \multicolumn{1}{c|}{0.001} & 0.001 & \multicolumn{1}{c|}{0}     \\
\multicolumn{1}{|c|}{}                      & \multicolumn{1}{c|}{8}  & 0.783 & \multicolumn{1}{c|}{0.006} & 0.193 & \multicolumn{1}{c|}{0.006} & 0.023 & \multicolumn{1}{c|}{0.003} & 0.002 & \multicolumn{1}{c|}{0.001} & 0.001 & \multicolumn{1}{c|}{0}     \\
\multicolumn{1}{|c|}{}                      & \multicolumn{1}{c|}{9}  & 0.715 & \multicolumn{1}{c|}{0.007} & 0.239 & \multicolumn{1}{c|}{0.007} & 0.042 & \multicolumn{1}{c|}{0.004} & 0.004 & \multicolumn{1}{c|}{0.001} & 0.001 & \multicolumn{1}{c|}{0}     \\
\multicolumn{1}{|c|}{}                      & \multicolumn{1}{c|}{10} & 0.642 & \multicolumn{1}{c|}{0.007} & 0.284 & \multicolumn{1}{c|}{0.007} & 0.065 & \multicolumn{1}{c|}{0.004} & 0.008 & \multicolumn{1}{c|}{0.002} & 0.001 & \multicolumn{1}{c|}{0}     \\ \hline
\multicolumn{1}{|c|}{\multirow{9}{*}{1000}} & \multicolumn{1}{c|}{2}  & 0.999 & \multicolumn{1}{c|}{0}     & 0.001 & \multicolumn{1}{c|}{0.001} &       & \multicolumn{1}{c|}{}      &       & \multicolumn{1}{c|}{}      &       & \multicolumn{1}{c|}{}      \\
\multicolumn{1}{|c|}{}                      & \multicolumn{1}{c|}{3}  & 0.998 & \multicolumn{1}{c|}{0.001} & 0.002 & \multicolumn{1}{c|}{0.001} &       & \multicolumn{1}{c|}{}      &       & \multicolumn{1}{c|}{}      &       & \multicolumn{1}{c|}{}      \\
\multicolumn{1}{|c|}{}                      & \multicolumn{1}{c|}{4}  & 0.997 & \multicolumn{1}{c|}{0.001} & 0.003 & \multicolumn{1}{c|}{0}     &       & \multicolumn{1}{c|}{}      &       & \multicolumn{1}{c|}{}      &       & \multicolumn{1}{c|}{}      \\
\multicolumn{1}{|c|}{}                      & \multicolumn{1}{c|}{5}  & 0.993 & \multicolumn{1}{c|}{0}     & 0.006 & \multicolumn{1}{c|}{0}     &       & \multicolumn{1}{c|}{}      &       & \multicolumn{1}{c|}{}      &       & \multicolumn{1}{c|}{}      \\
\multicolumn{1}{|c|}{}                      & \multicolumn{1}{c|}{6}  & 0.989 & \multicolumn{1}{c|}{0.001} & 0.011 & \multicolumn{1}{c|}{0.001} &       & \multicolumn{1}{c|}{}      &       & \multicolumn{1}{c|}{}      &       & \multicolumn{1}{c|}{}      \\
\multicolumn{1}{|c|}{}                      & \multicolumn{1}{c|}{7}  & 0.984 & \multicolumn{1}{c|}{0.001} & 0.015 & \multicolumn{1}{c|}{0.001} & 0.001 & \multicolumn{1}{c|}{0}     &       & \multicolumn{1}{c|}{}      &       & \multicolumn{1}{c|}{}      \\
\multicolumn{1}{|c|}{}                      & \multicolumn{1}{c|}{8}  & 0.979 & \multicolumn{1}{c|}{0.001} & 0.020 & \multicolumn{1}{c|}{0.001} & 0.001 & \multicolumn{1}{c|}{0}     &       & \multicolumn{1}{c|}{}      &       & \multicolumn{1}{c|}{}      \\
\multicolumn{1}{|c|}{}                      & \multicolumn{1}{c|}{9}  & 0.971 & \multicolumn{1}{c|}{0.001} & 0.028 & \multicolumn{1}{c|}{0.001} & 0.001 & \multicolumn{1}{c|}{0}     &       & \multicolumn{1}{c|}{}      &       & \multicolumn{1}{c|}{}      \\
\multicolumn{1}{|c|}{}                      & \multicolumn{1}{c|}{10} & 0.963 & \multicolumn{1}{c|}{0.001} & 0.036 & \multicolumn{1}{c|}{0.001} & 0.001 & \multicolumn{1}{c|}{0.001} &       & \multicolumn{1}{c|}{}      &       & \multicolumn{1}{c|}{}      \\ \hline
\end{tabular}
\caption{Overlap size distribution for the $d$-choice allocation instances created using the randomized construction process described in Appendix~\ref{subsec:approx_block_design_w_random_construction}.}
\label{table:overlap_distribution}
\end{table*}

\subsection{Proof of Theorem~\ref{thm:P_upper_bound_for_any_storage_allocation}}
\label{subsec:proof_thm_P_upper_bound_for_any_storage_allocation}
As described in Sec.~\ref{subsec:calculating_P_for_given_storage_allocation}, considering only a subset of the demand-vs-capacity conditions (in \eqref{eq:modified_cap_region_w_service_choice_spans}) gives an upper bound on $\mathcal{P}$.
Let us consider only one of the conditions for $t$ objects and choose these objects randomly, where $t$ is an integer in $[1, n]$. 
The span of service choices for these objects (total capacity available for jointly serving them) would then be the random variable $\mathrm{span}_t$ as defined in Def.~\ref{def:Pr_span}.
The cumulative capacity for these objects would be sufficient to keep the maximal load below $m$ only if 
\[
    \rho_1 + \dots + \rho_t \leq m \cdot \mathrm{span}_t.
\]
Then, the probability $\mathcal{P}_t$ that $\mathrm{span}_t$ is sufficient to meet the cumulative demand offered for them would be given as \eqref{eq:P_t}.
By the discussion above $\mathcal{P} < \mathcal{P}_t$.
We can then take the minimum of $\mathcal{P}_t$ for $t \in [1, n]$ and that would give us the upper bound in \eqref{eq:P_upper_bound_for_any_storage_allocation}.

\subsection{Proof of Theorem~\ref{thm:P_demand_is_bernoulli_lambda_eq_d}}
\label{subsec:proof_thm_P_demand_is_bernoulli_lambda_eq_d}
Demand of the active objects is set to $m d$; hence each active object fully uses all its service choices. This is why all pairs of active objects should have zero overlap in their service choices for the system to achieve the maximal load requirement.
We will set the maximal load $m = 1$ in the remainder of the proof for the sake of keeping the exposition clean. The proof can easily be extended for $m < 1$.

In the Theorem statement, the expression for $\mathcal{P}$ is presented in two equations: \eqref{eq:P_demand_is_bernoulli_lambda_eq_d_for_clustering_cyclic_block} given for clustering, cyclic or block design, and \eqref{eq:P_demand_is_bernoulli_lambda_eq_d_for_random} given for the random design.
As shown in the following, the derivation follows two different processes for the relevant storage constructions.

As the object demands are distributed as $d \times \mathrm{Bernoulli}(p)$, the number of active objects, denoted as $A$, is a random variable distributed as $\mathrm{Binomial}(n, p)$.
Let $a_i$ denote the objects selected to be active by the demand distribution.
In order to derive the expression for $\mathcal{P}$, we use the following steps:
(1) iterate over the active objects $a_i$ in order starting at $i = 1$,
(2) derive the probability $P_i$ that the active object $a_i$ does not overlap with any of the previous active objects $a_j$ for $j < i$,
(3) take the product of $P_i$ defined for all active objects.
Obviously, $P_1 = 1$ for any storage design.

\vspace{1ex}
\noindent
\textbf{Clustering, cyclic or block design}:
As discussed above, each time an object is selected to be active, all objects that overlap with it should be excluded from the active object selection. Let $c$ be the number of objects that need to be excluded each time an object is selected to be active. We have (a) $c = d$ for clustering design, (b) $c = 2d - 1$ for cyclic design, (c) $c = d^2 - d + 1$ for block design.

If the first $i - 1$ active objects have zero service choice overlap, they will lead to excluding $(i - 1) \cdot c$ objects from active object selection process.
Thus, the $i$th active object $a_i$ will not overlap with the previous objects with probability $(n - (i - 1) \cdot c) / (n - i + 1)$.
Since $A$ objects are selected to be active in total, we get the product of probabilities as stated in \eqref{eq:P_demand_is_bernoulli_lambda_eq_d_for_clustering_cyclic_block}.
The indicator function in the expression comes from the fact that it is impossible to select $A$ active objects without any service choice overlap if the number of objects $A \cdot c$ covered by them is greater than the total number of objects $n$.

\vspace{1ex}
\noindent
\textbf{Random design}:
In random design, for each active object, $d$ nodes are selected at random and their capacity is fully used to serve the demand for the object. Hence, for the system to meet the maximal load requirement, no other active object should select any node that has been selected before.
Notice that this rationale differs from the one we used for the previous storage designs. In this case, we focus on nodes being excluded by the active objects rather than other objects.

If the first $i - 1$ objects are selected to have zero overlaps in their service choices, their selection will exclude (cover) $(i - 1) \cdot d$ nodes.
Thus, the nodes selected for the $i$th active object $a_i$ will not overlap with the previous objects with probability ${n - (i - 1) \cdot d \choose d} / {n \choose d}$.
Taking the product of these probabilities for all $A$ active objects gives \eqref{eq:P_demand_is_bernoulli_lambda_eq_d_for_random}.
The indicator function in the expression comes from the fact that, in this case, it is not possible to select $A$ active objects without any service choice overlap if the number of nodes $A \cdot d$ covered by them is greater than the total number of nodes $n$.

\subsection{Proof of Lemma~\ref{lm:P_clustering}}
\label{subsec:proof_lm_P_clustering}
Notice that in a storage allocation with clustering design, clusters are decoupled from each other in terms of the storage nodes and the object demands. This is why we can think of each cluster as an independent sub-system. 

There are $n / d$ clusters in the system ($d|n$).
Let us use $\mathcal{P}_c$ to denote the probability that a cluster meets the maximal load requirement. Given that clusters are identical and they are independent from each other, we can express the probability that the entire system meets the maximal load requirement as
\begin{equation}
    \mathcal{P} = \left(\mathcal{P}_c\right)^{n / d}.
\label{eq:P_w_P_c}
\end{equation}

Each cluster stores $d$ objects and each object is hosted on each node within the cluster. Thanks to this, the demand for any of the $d$ objects can be served at any node within the cluster. This implies that a cluster can meet the maximal load requirement as long as the cumulative demand offered on the cluster is less than the cumulative capacity available within the cluster. That is
\[
    \mathcal{P}_c = \Prob{\rho_1 + \dots + \rho_d \leq m \cdot d}.
\]
Substituting this expression of $\mathcal{P}_c$ in \eqref{eq:P_w_P_c} gives us \eqref{eq:P_clustering}.


\subsection{Proof of Theorem~\ref{thm:P_bounds_for_clustering}}
\label{subsec:proof_thm_P_bounds_for_clustering}
As given in the Theorem statement, object demands $\rho_i$ are sub-gaussian in the sense that there exists constants $c, C > 0$ such that
\begin{equation}
    \exp(c t^2) \leq \phi_{\rho_i}(t) \leq \exp(C t^2), \qquad \forall t > 0,
\label{eq:rho_i_sub_gaussian_inequality}
\end{equation}
where $\phi_{\rho_i}(t)$ is the moment generating function of $\rho_i$.

Let us modify the object demands and define their \emph{zero-mean} versions as $\tilde{\rho}_i = \rho_i - \mu$ for $\mu = \Exp{\rho_i}$. Given that $\rho_i$ are sub-gaussian in the sense given in \eqref{eq:rho_i_sub_gaussian_inequality} and the moment generating function of $\tilde{\rho_i}(t)$ is given as $\phi_{\tilde{\rho_i}}(t) = \exp(-\mu t) \phi_{\rho_i}(t)$, there exists $c^\prime, C^\prime > 0$ such that
\begin{equation}
    \exp(c^\prime t^2) \leq \phi_{\rho_i}(t) \leq \exp(C^\prime t^2), \qquad \forall t > 0.
\label{eq:tilde_rho_i_sub_gaussian_inequality}
\end{equation}
That is, $\tilde{\rho}_i$ are sub-gaussian in the same sense that $\rho_i$ are sub-gaussian.
This, together with the fact that $\tilde{\rho}_i$ are zero-mean makes it possible to use Theorem 3 in \cite{NonAsymptoticLowerTailBounds:ZhangZ20} and find the following bounds. There exists constants $\alpha, \beta, \gamma > 0$ such that the sum $\tilde{S} = \tilde{\rho}_1 + \dots + \tilde{\rho}_d$ satisfies
\[
    \exp(-\alpha x^2 / d) \leq 1 - \tilde{F}_d(x) \leq \gamma \exp(-\beta x^2 / d), \quad \forall x \geq 0.
\]
where $\tilde{F}_d(x)$ denotes the distribution function for $\tilde{S}$.

Let us now define the sum $S = \rho_1 + \dots + \rho_d$ and its distribution function as $F_d(x)$. As shown in Lemma~\ref{lm:P_clustering}, each cluster in the storage allocation meets the maximal load $m$ with probability $F_d(m \cdot d)$.
Given that $\tilde{\rho}_i = \rho_i - \mu$, we have
\[
    F_d(m \cdot d) = \tilde{F}_d((m - \mu) \cdot d).
\]
Then, using the bounds given above for $\tilde{F}_d$, we obtain the following bounds on $F_d(m \cdot d)$
\begin{equation*}
\begin{split}
    &1 - \gamma \exp(-\beta (m - \mu)^2 \cdot d)  \\
    &\leq F_d(m \cdot d) \\
    &\leq 1 - \exp(-\alpha (m - \mu)^2 \cdot d),
    \quad \forall m \geq \mu.
\end{split}
\end{equation*}
Substituting these bounds in \eqref{eq:P_clustering} gives us \eqref{eq:P_bounds_for_clustering}.

\subsection{Proof of Corollary~\ref{cor:P_for_clustering_as_d_to_infty}}
\label{subsec:proof_cor_P_for_clustering_as_d_to_infty}
Let us suppose that $d = c\log(n)^\alpha$ for constants $c, \alpha > 0$.
Both the lower and upper bound in \eqref{eq:P_bounds_for_clustering} converge to $1$ as $n \to \infty$ if $\alpha \geq 1$, and the bounds converge to $0$ if $0 < \alpha < 1$.

\subsection{Proof of Lemma~\ref{lm:P_upper_bound_for_any_design}}
\label{subsec:proof_lm_P_upper_bound_for_any_design}
System can possibly meet the maximal load requirement only if a set of objects $O$ has a cumulative demand less than its scaled span $m \cdot \mathrm{span}(O)$.
In a $d$-choice allocation, when objects in $O$ have zero overlap in their service choices, they would attain the maximum span, which is given as $\card{O} \cdot md$.
This implies that it is necessary for the cumulative demand for any $s$ consecutive objects to be less than $s \cdot md$ in order for the system to possibly meet the maximal load requirement.
This necessary condition gives us the following upper bound
\[
    \mathcal{P} \leq \Prob{S_s^{(c)} \leq s \cdot m \; d}.
\]
This bound holds for any $1 \leq s \leq n$. Also, the minimum of these bounds over $s$ should also be an upper bound. This gives us \eqref{eq:P_upper_bound_for_any_design}.

\subsection{Convergence of circular scan statistic to its non-circular counterpart}
\label{subsec:circular_scan}
%
We here show that the circular $s$-scan statistic $S_s^{(c)}$ defined for $n$ i.i.d. samples (see Def.~\ref{def:scan_statistic}) converge to its non-circular counterpart $S_s$.
Note that the proofs presented in this section are mostly translation of those given for circular maximal $d$-spacing in \cite{LoadBalancing:AktasFS21}[Appendix~E] to the scan statistic we make use of in this paper. This translation is possible after seeing that maximal spacing is a specialized scan statistic for the spacings between the ordered uniform samples on the unit line.
We here present the proofs for scan statistic for completeness.

In the following, we show convergence first in distribution, then in probability, and finally almost surely.
Note that showing almost sure convergence implies convergence in probability, which then implies convergence in distribution.
We present the convergence in this order for a more streamlined exposition of the arguments.

\begin{lemma}
  For $d < n$,
  \begin{equation*}
    \Prob{S_s > x} \leq \Prob{S_s^{(c)} > x} \leq \frac{n}{n-d} \Prob{S_s > x}.
  \end{equation*}
\label{eq:Pr_Sc_leq_scaled_Pr_S}
\end{lemma}
\begin{proof}
  Let us denote the events $\left\{S_s > x\right\}$ and $\left\{S_s^{(c)} > x\right\}$ respectively with $R$ and $C$.
  
  The first inequality is immediate; if a sequence of random samples $\bm{x} = (x_1, x_2, \dots, x_n) \in R$ then $\bm{x} \in C$, while the opposite direction may not hold. Thus, $R \subseteq C$, hence $\Prob{R} \leq \Prob{C}$.
  
  Next we show the second inequality.
  Let a random sample sequence $\bm{x} \in R$. Then, at least $n - d$ different permutations of $\bm{x}$ lie in $R$. In order to see this, let the maximum $s$-scan within $\bm{x}$ be $\bm{y} = (x_i, \dots, x_{i+s-1})$. Shifting (by feeding what is shifted out back in the sequence at the opposite end) $\bm{x}$ to the left by at most $i - 1$ times will preserve $\bm{y}$, hence each of the $i - 1$ shifted versions will also lie in $R$. Similarly, shifting $\bm{x}$ to the right by at most $n - (i+s-1)$ times will also preserve $\bm{y}$.
  We call such permutations, which are obtained by shifting with wrapping around, a \textit{cyclic permutation}.
  
  Let us introduce a set $R^\prime \subset R$ such that for any $\bm{x} \in R^\prime$, no cyclic permutation of $\bm{x}$ lies in $R^\prime$. $R$ contains at least $n - s$ cyclic permutations of every $\bm{x} \in R^\prime$. This together with the fact that probability of sampling a sequence $\bm{x}$ is independent of the order of the samples gives us $(n - d) \Prob{R^\prime} \leq \Prob{R}$.
  
  Now let $\bm{x}^\prime \in C$. All $n - 1$ cyclic permutations of $\bm{x}^\prime$ will also lie in $C$ (recall that we are now working in the circular setting).
  This together with the fact $R^\prime \subset R \subseteq C$, and the fact that probability of sampling a sequence $\bm{x}$ is independent of the order of the samples gives us $n \cdot \Prob{R^\prime} = \Prob{C}$. 
  Putting it all together, we have $\Prob{C} / n = \Prob{R^\prime} \leq \Prob{R} / (n - d)$, which yields the second inequality.
  
\end{proof}

\begin{lemma}
  For $d = o(n)$, $S_s^{(c)} / S_s \to 1$ in probability as $n \to \infty$.
\label{lm:S_circular_convergence_in_prob}
\end{lemma}
\begin{proof}
  It is easy to see $S_s^{(c)} \geq S_s$.
  Let $D = S_s^{(c)} - S_s$ and $S$ be the set of all sequence of spacings for which $D > 0$.
  For every $\bm{x} \in S$, $d - 1$ of its cyclic permutations (see the Proof of Lemma~\ref{eq:Pr_Sc_leq_scaled_Pr_S} for the definition of a cyclic permutation) also lie in $S$ while the remaining $n - d$ of them lie in $S^c$ (complement of $S$).
  Thus, for every $d$ points in $S$, there are at least $n - d$ points in $S^c$, and all the points in $S$ or $S^c$ have the same probability measure in both sets.
  This gives us the following upper bound $\Prob{D > 0} = \Prob{S} \leq d / n$, which $\to 0$ as $n \to \infty$. This implies $S_s^{(c)} / S_s \to 1$ in probability.
\end{proof}
In order to use the results known for the convergence of $S_s$ in probability or a.s. in addressing $S_s^{(c)}$, we need the following Lemma.

\begin{lemma}
  For $d = o(n)$, $S_s^{(c)} / S_s \to 1$ a.s. as $n \to \infty$.
\label{lm:S_circular_to_S_as}
\end{lemma}

We here skip the proof for Lemma~\ref{lm:S_circular_to_S_as} as the previously presented Lemma~\ref{eq:Pr_Sc_leq_scaled_Pr_S} and ~\ref{eq:Pr_Sc_leq_scaled_Pr_S} are sufficient for this paper.
We refer the interested reader to \cite{LoadBalancing:AktasFS21}[Appendix E] for the full exposition showing that maximal $d$-spacing on the unit circle converges to its counterpart on the unit line almost surely. As for the proofs presented above for Lemma~\ref{eq:Pr_Sc_leq_scaled_Pr_S} and ~\ref{eq:Pr_Sc_leq_scaled_Pr_S}, the proof given in \cite{LoadBalancing:AktasFS21} for almost sure convergence can be modified similarly to show that the circular scan statistic converges to its non-circular counterpart almost surely as $n \to \infty$.

\subsection{Proof of Lemma~\ref{lm:P_bounds_for_rgap}}
\label{subsec:proof_lm_P_bounds_for_rgap}
%
\textbf{Upper bound:}
The system cannot meet the maximal load requirement if a set of objects $O$ has a cumulative demand larger than its scaled span $m \cdot \mathrm{span}(O)$.
Lemma~\ref{lm:on_rgap} states that the span of every $i$ consecutive objects (with respect to their indices) is at most $i + 2r$. This implies that it is necessary for the cumulative demand for any $i$ consecutive objects to be less than $i + 2r$ in order for the system to possibly meet the maximal load requirement.
This necessary condition shows that the right-hand side 
of the inequality in \eqref{eq:P_bounds_for_rgap} is an upper bound on $\mathcal{P}$.

\vspace{1ex}
\noindent
\textbf{Lower bound:}
Let $x$ be an integer in $[1, n]$.
Consider the spiky load scenario starting at object $o_x$ where demand $\rho_i$ is $m \cdot d$ when $i = x + (r + 1)j$ for $j=0, 1, \dots, \floor{n / (r + 1)}$, and $0$ otherwise.
In this case, a demand of magnitude $m \cdot d$ for each spiky object $o_i$ can be supplied by using up the capacity in all its $d$ service choices. This is because, by the $r$-gap design property, all other objects that overlap with a spiky object (in service choices) have zero demand.
System can meet the maximal load requirement while serving the spiky load regardless of the value for $x$.
Given that system's service capacity region is convex (see Sec.~\ref{sec:perf_metric}), the system can serve all convex combinations of demand vectors that correspond to any set of spiky load scenarios.
This can be expressed as follows: the system can meet the maximal load requirement as long as the cumulative demand on every $r + 1$ consecutive objects is at most $d$.
This implies that the left-hand side of the inequality in \eqref{eq:P_bounds_for_rgap} is a lower bound on the actual value of $\mathcal{P}$.

\subsection{Proof of Lemma~\ref{lm:P_bounds_for_cyclic}}
\label{subsec:proof_lm_P_bounds_for_cyclic}
Lower bound come from substituting $r = d - 1$, which is the lowest possible value for $r$.
Upper bound comes from 
(1) observing in cyclic design that the span of every $s$ consecutive objects is at most $s + d - 1$,
(2) substituting these values in the upper bound given in Lemma~\ref{lm:P_bounds_for_rgap}.

\subsection{Proof of Theorem~\ref{thm:P_asymptotic_bounds_for_cyclic}}
\label{subsec:proof_thm_P_asymptotic_bounds_for_cyclic}
Recall the definition of \emph{scan statistic} $S_s$ from Def.~\ref{def:scan_statistic}.
When the object demands $\rho_i$ are non-zero (positive), we can use Theorem 2, \cite{PoissonApproxForScanStat:DemboK92} to conclude that
\[
    \Prob{S_s \leq x} = \exp\left((n - s + 1)Q_s(x)\right)
\]
in the limit $n \to \infty$.
Recall that, as shown in Appendix~\ref{subsec:circular_scan}, circular scan statistic $S_s^{(c)}$ converges to its non-circular counterpart $S_s$ almost surely as $n \to \infty$.
This implies that the asymptotic expression given above for the distribution of $S_s$ is also valid for the distribution of its circular counterpart $S_s^{(c)}$.
Substituting this expression of $S_s^{(c)}$ in \eqref{eq:P_bounds_for_cyclic}, we obtain the asymptotic bounds for $\mathcal{P}$ as given in \eqref{eq:P_asymptotic_bounds_for_cyclic}.

\subsection{Proof of Corollary~\ref{cor:P_asymptotic_bounds_for_cyclic_simpler}}
\label{subsec:proof_cor_P_asymptotic_bounds_for_cyclic_simpler}
The lower bound directly follows from the lower bound given in \eqref{eq:P_asymptotic_bounds_for_cyclic}.
We find the upper bound as follows. 
Among the upper bounds over which we take the minimum in \eqref{eq:P_asymptotic_bounds_for_cyclic}, let us consider only the one for $s = d$. This gives us
\[
    \mathcal{P} \leq \exp\left(-w_{n, d} \; Q_d(m(2d - 1))\right).
\]
Given that $Q_d(x)$ is non-increasing in $x$, we have
\[
    Q_d(m(2d - 1)) \geq Q_d(2md).
\]
This implies the following for the upper bound given in \eqref{eq:P_asymptotic_bounds_for_cyclic}
\[
    \mathcal{P} \leq \exp\left(-w_{n, d} \; Q_d(m(2d - 1))\right) \leq \exp\left(-w_{n, d} \; Q_d(2md)\right).
\]

\subsection{Proof of Corollary~\ref{cor:P_asymptotic_bounds_for_cyclic_insightful}}
\label{subsec:proof_cor_P_asymptotic_bounds_for_cyclic_insightful}
In the proof of Theorem~\ref{thm:P_bounds_for_clustering}, we showed the following.
Suppose that the objects demands $\rho_i$ are sub-gaussian in the sense that there exists constants $c, C > 0$ such that
\begin{equation*}
    \exp(c t^2) \leq \phi_{\rho_i}(t) \leq \exp(C t^2), \qquad \forall t > 0.
\end{equation*}
Then, there exists constants $\alpha, \beta, \gamma > 0$ such that the tail distribution $Q_d$ satisfies
\begin{equation*}
  \exp(-\alpha x^2 / d) \leq Q_d(x + d\mu) \leq \gamma \exp(-\beta x^2 / d), \quad \forall x \geq 0.
\end{equation*}

The lower and upper bound given for $\mathcal{P}$ in \eqref{eq:P_asymptotic_bounds_for_cyclic_simpler} are expressed in terms of $Q_d(md)$ and $Q_d(2md)$.
Using the bounds given above for $Q_d$, we find the following upper bound for $Q_d(md)$ and lower bound for $Q_d(2md)$
\begin{equation*}
\begin{split}
  Q_d(md) &\leq \gamma \exp(-d \cdot \beta(m - \mu)^2) \\
  Q_d(2md) &\geq \exp(-d \cdot \alpha(2m - \mu)^2).
\end{split}
\end{equation*}
Substituting these bounds for $Q_d(md)$ and $Q_d(2md)$ respectively in \eqref{eq:P_asymptotic_bounds_for_cyclic_simpler} gives us the bounds in \eqref{eq:P_asymptotic_bounds_for_cyclic_insightful}.


\subsection{Proof of Corollary~\ref{cor:P_for_cyclic_as_d_to_infty}}
\label{subsec:proof_cor_P_for_cyclic_as_d_to_infty}
Let us suppose that $d = c\log(n)^\alpha$ for constants $c, \alpha > 0$.
Both the lower and upper bound in \eqref{eq:P_asymptotic_bounds_for_cyclic_insightful} go to $1$ as $n \to \infty$ if $\alpha \geq 1$, and the bounds go to $0$ if $0 < \alpha < 1$.

\subsection{Proof of Lemma~\ref{lm:P_bounds_for_block_design}}
\label{subsec:proof_lm_P_bounds_for_block_design}
Recall from Sec.~\ref{subsec:balanced_d_choice} that, in a storage allocation with block design, every pair of objects overlaps at \emph{exactly} one node in their service choices. We here refer to this fact as \textbf{F}.

\noindent
\textbf{Upper bound:}\space
Let us consider a set of $d$ objects, which we denote as  $O$, that are arbitrarily chosen from all objects. 
As discussed previously, the total capacity to jointly serve a set of objects grows with their span. 
Given the fact \textbf{F}, all objects in $O$ overlap at the same node in the best case. This will give us the maximum possible span as $\mathrm{span}(O) = d(d-1)$.
Recall also that system cannot meet the maximal load requirement if a set of objects $O$ has a cumulative demand larger than its scaled span $m \cdot \mathrm{span}(O)$. These two observations give us the following necessary condition: cumulative demand for any set of $d$ objects must be less than $m \cdot d(d-1)$ to meet the maximal load requirement.
Instead of considering all subsets of $d$ objects, we can just consider the subsets of $d$ consecutive objects with respect to their indices, i.e., $\{o_i, \dots, o_{i + d - 1}\}$ where indices are implicitly defined as $i \mod n$. This makes it possible to state the necessary condition in terms of the circular scan statistic (see Def.~\ref{def:scan_statistic}) as $S_s^{(c)} \leq m \cdot (d^2 - d)$
This gives us the upper bound in \eqref{eq:P_bounds_for_block_design}.

\noindent
\textbf{Lower bound:}\space
Similar to the proof for the upper bound, let us consider a set $O$ of $d$ arbitrarily chosen objects.
Given the fact \textbf{F}, objects in $O$ overlap at different nodes in the worst case for $\mathrm{span}(O)$. This will give us the minimum possible span as $\mathrm{span}(O) = d(d-1)/2$.

We here consider the spiky load scenario discussed in the proof of Lemma~\ref{lm:P_bounds_for_rgap}; let $x$ be an integer in $[0, n]$, and the offered load for $o_i$ is $\rho$ if $i = x + d j$ for some $j = 0, 1, \dots, \floor{n/d}$ and $0$ otherwise.
Let us refer to objects with spiky load as ``a spiky object''.
As discussed in the previous paragraph, the worst-case sharing for jointly serving the objects is when a spiky object has to share $d - 1$ of its service choices with other spiky objects.
In the worst case, system can meet the maximal load requirement if $\rho \leq m \cdot (1 + (d - 1)/2)$.
Given that the system's service capacity region is convex (see Sec.~\ref{sec:perf_metric}), the system can serve all convex combinations of any set of spiky load scenarios.
This can be expressed as follows: the system can meet the maximal load requirement as long as the cumulative demand on every $d$ consecutive objects is at most $m \cdot (1 + (d - 1)/2)$.
This gives us the lower bound for $\mathcal{P}$ as stated in \eqref{eq:P_bounds_for_block_design}.

\subsection{Proof of Theorem~\ref{thm:P_asymptotic_bounds_for_block}}
\label{subsec:proof_thm_P_asymptotic_bounds_for_block}
The proof follows the same sequence of arguments used in the proof for Theorem~\ref{thm:P_asymptotic_bounds_for_cyclic}, which is presented in Appendix~\ref{subsec:proof_thm_P_asymptotic_bounds_for_cyclic}.

\subsection{Proof of Corollary~\ref{cor:P_asymptotic_bounds_for_block_insightful}}
\label{subsec:proof_cor_P_asymptotic_bounds_for_block_insightful}
In the proof of Theorem~\ref{thm:P_bounds_for_clustering}, we showed the following.
Suppose that the objects demands $\rho_i$ are sub-gaussian in the sense that there exists constants $c, C > 0$ such that
\begin{equation*}
    \exp(c t^2) \leq \phi_{\rho_i}(t) \leq \exp(C t^2), \qquad \forall t > 0.
\end{equation*}
Then, there exists constants $\alpha, \beta, \gamma > 0$ such that the tail distribution $Q_d$ satisfies
\begin{equation*}
  \exp(-\alpha x^2 / d) \leq Q_d(x + d\mu) \leq \gamma \exp(-\beta x^2 / d), \quad \forall x \geq 0.
\end{equation*}

The lower and upper bound given for $\mathcal{P}$ in \eqref{eq:P_asymptotic_bounds_for_block} are expressed in terms of $Q_d(md/2)$ and $Q_d(m(d^2 - d))$.
Using the bounds given above for $Q_d$, we find the following upper bound for $Q_d(md/2)$ and lower bound for $Q_d(m(d^2 - d))$
\begin{equation*}
\begin{split}
  Q_d(md/2) &\leq \gamma \exp(-d \cdot \beta(m/2 - \mu)^2), \\
  Q_d(d(m(d - 1) - \mu) &\geq \exp(-d \cdot \alpha(m(d - 1) - \mu)^2).
\end{split}
\end{equation*}
Substituting these bounds for $Q_d(md/2)$ and $Q_d(m(d^2 - d))$ respectively in \eqref{eq:P_asymptotic_bounds_for_block} gives us the bounds in \eqref{eq:P_asymptotic_bounds_for_block_insightful}.

\subsection{Proof of Lemma~\ref{lm:P_upper_bound_for_random}}
\label{subsec:proof_lm_P_upper_bound_for_random}
As discussed many times in the paper so far, given a set of $u$ objects $O$, it is necessary for their cumulative demand to be at most equal to the cumulative capacity available to jointly serve them. 
With random design, service choices (nodes) for each object are selected at random without replacement. Hence, the span of $O$ is given by the occupancy metric for random allocation with $d$-complexes $N_{n, d, u}$ as defined in Def.~\ref{def:occupany_of_random_allocation_w_complexes}.
We then have the probability of this necessary condition given as
\begin{equation}
    \Prob{\rho_1 + \dots + \rho_u \leq m \cdot N_{n, d, u}} = F_u\left(m \cdot N_{n, d, u}\right).
\label{eq:Pr_cum_demand_for_u_objs_leq_cum_cap}
\end{equation}

Let us now partition the objects into a collection of $v$ subsets where the $i$th set within the collection contains $u_i$ objects for $i = 1, \dots, v$, and $u_1 + u_2 + \dots + u_v = n$.
The necessary condition given above must hold for each of the subsets that constitute the partition.
Given that the service choices for the objects are chosen independently, probability that the necessary condition will hold jointly for all the partition subsets is given by the product of the probabilities for individual subsets.
Probability of the joint necessary condition is given by
\[
    \prod_{i = 1}^{v} \E_{N_{n, d, u_i}}\left[F_{u_i}(m \cdot N_{n, d, u_i})\right]
\]
which serves as the upper bound given in \eqref{eq:P_upper_bound_for_random}.

\subsection{Proof of Corollary~\ref{cor:P_upper_bound_for_random_w_even_partitioning}}
\label{subsec:proof_cor_P_upper_bound_for_random_w_even_partitioning}
We obtain \eqref{eq:P_upper_bound_for_random_w_even_partitioning} by setting $u_i = n / u$ for all $i$ in \eqref{eq:P_upper_bound_for_random}.

\subsection{Proof of Lemma~\ref{lm:P_lower_bound_for_constrained_random}}
\label{subsec:proof_lm_P_lower_bound_for_constrained_random}
%
Let $x$ be an integer in $[1, n]$.
Consider the spiky load scenario starting at object $o_x$ where demand $\rho_i$ equals $\rho$ when $i = x + dj$ for $j=0, 1, \dots, \floor{n / d}$, and $0$ otherwise.
In this case, the spiky demand can be served as long as $\rho \leq md / v_{\max}$. This is because, the constraint defined on the number of overlapping $d$-hop siblings dictates that each spiky object can overlap with at most $v_{\max}$ other spiky objects at its service choices.
Note also that all other non-spiky objects that overlap with a spiky object (in service choices) have zero demand.

System can meet the maximal load requirement while serving the spiky load regardless of the value for $x$.
We can then follow the same steps outlined in the proof of the lower bound given in Lemma~\ref{lm:P_bounds_for_rgap} and show the lower bound in \eqref{eq:P_lower_bound_for_constrained_random}.

\end{document}